\documentclass[10pt,a4paper]{article}
\usepackage[utf8]{inputenc}
\usepackage{amsmath, amsthm}
\usepackage{amsfonts}
\usepackage{amssymb}
\usepackage{hyperref}	

\usepackage[left=1cm,right=1cm,top=1cm,bottom=1.2cm]{geometry}

\usepackage{enumerate}

\usepackage{enumerate}

\theoremstyle{plain}
\newtheorem{thm}{Theorem}[section]

\newtheorem{lem}[thm]{Lemma}

\newtheorem{prop}[thm]{Proposition}
\newtheorem*{propunb}{Proposition}

\newtheorem{cor}[thm]{Corollary}
\newtheorem*{corunb}{Corollary}

\newtheorem{conjecture}{Conjecture}[section]
\newtheorem*{thmunb}{Theorem}

\newtheorem{theo}{Theorem}
\newtheorem{theotwo}{Theorem}

\theoremstyle{remark}
\newtheorem{rmk}{Remark}

\newtheorem*{rmks}{Remarks}

\theoremstyle{definition}

\newtheorem{defn}{Definition}[section]
\newtheorem*{defunb}{Definition}

\newcommand{\RR}{\mathbb{R}}

\theoremstyle{plain}

\theoremstyle{remark}

\theoremstyle{definition}

\newcommand{\A}{\mathcal{A}}
\newcommand{\R}{\mathcal{R}}
\newcommand{\T}{\mathcal{T}}
\newcommand{\uA}{u_{\mathcal{A}^+}(v)}

\newcommand{\CH}{\mathcal{CH}_{i^+}}

\newcommand{\uend}{u(\CH \cup \mathcal{S}_{i^+})}

\usepackage{bm}

\usepackage{color}

\usepackage{float}

\addtocounter{tocdepth}{-2}
\usepackage{graphicx}
\usepackage{authblk}

\numberwithin{equation}{section}
\begin{document}
	\title{The breakdown of weak null singularities inside black holes}
	\author[1]{Maxime Van de Moortel \thanks{E-mail : mv715@rutgers.edu}}  
	
	\affil[1]{Rutgers University, Department of Mathematics,
		110 Frelinghuysen Road
		Piscataway, NJ 08854, USA} 	\date{\vspace{-8ex}}
	\maketitle
	\abstract 
	It is widely expected that generic black holes have a non-empty but \textit{weakly singular} Cauchy horizon, due to mass inflation. Indeed this has been proven by the author in the spherical collapse of a charged scalar field, under decay assumptions of the field in the black exterior which are conjectured to be generic. A natural question then arises: can this weakly singular Cauchy horizon close off the space-time, or does the weak null singularity \underline{necessarily} ``break down'', giving way to a different type of singularity? The main result of this paper is to prove that the Cauchy horizon \underline{cannot} ever ``close off'' the space-time. As a consequence, the weak null singularity breaks down and transitions to a stronger singularity for which the area-radius $r$ extends to $0$. 

 	\section{Introduction}
 	
 	The characterization of singularities inside black holes is a fundamental problem in General Relativity, which is related to the fate of in-falling observers and the very validity of the principle of determinism. ``Strong'' singularities, for which the area-radius $r$ extends to $0$, are already present in the interior of the Schwarzschild black hole and raised immense interest, see \cite{spacelike1} and its vast developments, and for instance the review \cite{spacelike2}. 
 	
 	For many years, it was believed that a generic black hole interior is necessarily delimited by a singularity which is everywhere strong and space-like. It is now well-understood that the above belief is \underline{false}. Indeed, it has been proven in \cite{KerrStab} that all dynamical black holes settling down to Kerr possess a Cauchy horizon in the black hole interior, i.e.\ a null boundary spanned by spheres of non-zero radius. Additionally, a Cauchy horizon on which $r>0$  necessarily occurs for dynamical charged black holes settling down to Reissner--Nordstr\"{o}m in spherical symmetry, as proven in  \cite{MihalisPHD,Moi}; at the heuristic level, this phenomenon is explained by the repulsive mechanism provided by the Maxwell field in the charged case, or provided by angular momentum in the non spherically-symmetric case.
 	
 	While  the future \color{black} boundary components on which $r=0$ are (strong) singularities, in contrast Cauchy horizons are  \textit{not always} singular: for instance the Reissner--Nordstr\"{o}m Cauchy horizon is smoothly regular. Nevertheless, the pioneering works \cite{Hiscock,Ori,Poisson,PoissonIsrael} suggested 
 	that the Cauchy horizon of \underline{generic} dynamical black holes is, in fact, a \textit{weak null singularity}. This phenomenon is known as ``mass inflation''. It has been proven indeed c.f.\ \cite{Mihalis1,JonathanStab,Moi,Moi4} that the Cauchy horizon of charged spherically symmetric black holes is weakly singular, under assumptions on the black hole exterior which are conjectured to be generic; retrieving these assumptions remains an important open problem. Weak null singularities have  been constructed in vacuum \cite{JonathanWeakNull} and are conjecturally present in the interior of generic perturbations of Kerr black holes. 

 	In view of the mathematical evidence in favor of the generic character of weak null singularities, the \textit{necessity} of the occurrence in collapse of $r=0$ singularities on part of the boundary, ironically, becomes subject to questioning: do weak null singularities necessarily ``break down'' in finite retarded time, and a new type of (presumably stronger) singularity takes over? \color{black} Or, on the contrary, is it possible in some cases that they subsist up to the center of symmetry and close off the space-time, as depicted in the Penrose diagram of Figure \ref{Fig3} ? This is not a moot point, since in the \textit{two-ended} asymptotically flat case, the Cauchy horizon closes-off the space-time  for a large class of dynamical solutions, as depicted in Figure~\ref{Figure5}, see also section~\ref{twoended}.


 	
 	In the present paper, we carry out the first global study of the black hole interior what can be seen as \color{black} the simplest model in which this question makes sense, namely the gravitational collapse of a charged scalar field, governed by the Einstein--Maxwell--Klein--Gordon equations in spherical symmetry: \begin{equation} \label{1.1}   Ric_{\mu \nu}(g)- \frac{1}{2}R(g)g_{\mu \nu}= \mathbb{T}^{EM}_{\mu \nu}+  \mathbb{T}^{KG}_{\mu \nu} ,    \end{equation} 
 	\begin{equation} \label{2.1} \mathbb{T}^{EM}_{\mu \nu}=2\left(g^{\alpha \beta}F _{\alpha \nu}F_{\beta \mu }-\frac{1}{4}F^{\alpha \beta}F_{\alpha \beta}g_{\mu \nu}\right),
 	\end{equation}
 	\begin{equation} \label{3.1} \mathbb{T}^{KG}_{\mu \nu}= 2\left( \Re(D _{\mu}\phi \overline{D _{\nu}\phi}) -\frac{1}{2}(g^{\alpha \beta} D _{\alpha}\phi \overline{D _{\beta}\phi} + m ^{2}|\phi|^2  )g_{\mu \nu} \right), \end{equation} \begin{equation} \label{4.1} \nabla^{\mu} F_{\mu \nu}= \frac{ q_{0} }{2}i (\phi \overline{D_{\nu}\phi} -\overline{\phi} D_{\nu}\phi) , \; F=dA ,
 	\end{equation} \begin{equation} \label{5.1} g^{\mu \nu} D_{\mu} D_{\nu}\phi = m ^{2} \phi , 	\end{equation} 	which feature a scalar field $\phi$ of charge $q_0 \neq 0$ and of mass $m^2 \geq 0$ and $D_{\mu}= \nabla_{\mu}+iq_0 A_{\mu}$ is the gauge derivative. This model has been extensively studied c.f.\  \cite{HodPiran1,HodPiran2,Kommemi,KonoplyaZhidenko,OrenPiran}. In the present paper, we consider solutions with one-ended asymptotically flat initial data, diffeomorphic to $\RR^3$, as they model black holes arising from gravitational collapse.

 	Our main result in this context is summarized in layman's terms as such:
 	\begin{thmunb} In the spherical  collapse of a charged scalar field, weak null singularities \underline{necessarily} break down.
 	\end{thmunb}

 	Therefore, a weakly singular Cauchy horizon can never close off the space-time, so the Penrose diagram of Figure \ref{Fig3} is ruled out in the presence of a weak null singularity.
 		\begin{figure}[H]
 		
 		\begin{center}
 			
 			\includegraphics[width=68 mm, height=60 mm]{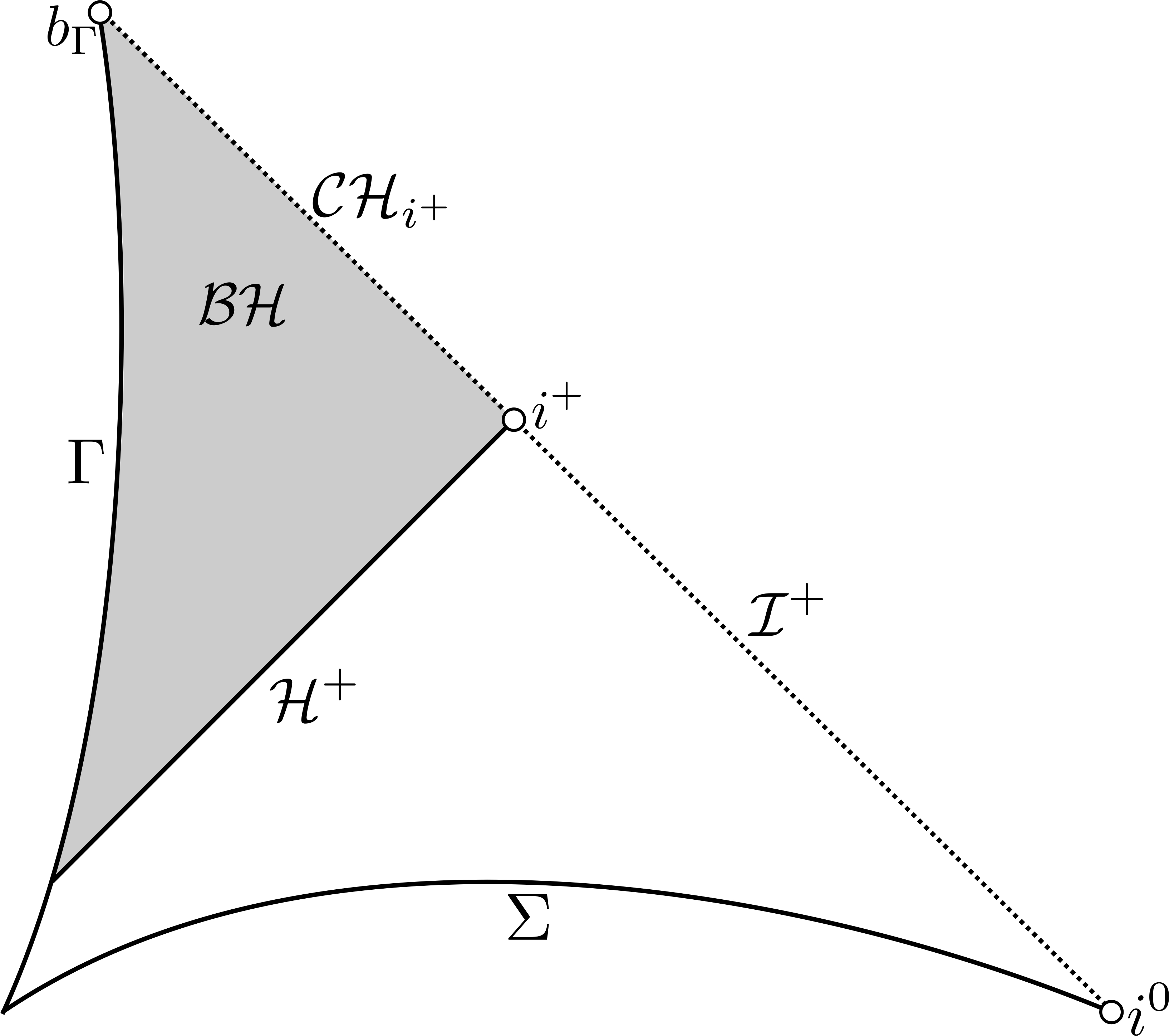}
 			
 		\end{center}
 		
 		\caption{Penrose diagram whose existence we disprove if $\CH$ is weakly singular.}
 		\label{Fig3}
 	\end{figure} 
 	As a consequence of the necessary breakdown of weak null singularities, we obtain a proof of the ``$r=0$ singularity conjecture'' \cite{Kommemi}: an one-ended black hole which does not have a ``locally naked singularity'' always features an ``$r=0$ singularity'', in addition to a weakly singular Cauchy horizon, and its Penrose diagram is given by Figure \ref{Fig2}. To prove this result,  we make use of an earlier classification of possible Penrose diagrams \cite{Kommemi}.  As a consequence, we show there exists a so-called ``first singularity'' where $r=0$, i.e.\ a Terminal Indecomposable Past (associated to a boundary point) whose closure has compact intersection with the Cauchy initial hypersurface.


 	Our approach fundamentally uses a contradiction argument. More precisely, we assume that the Penrose diagram is given by Figure \ref{Fig3}, where $\CH$ is a Cauchy horizon subject to mass inflation; by this, we mean that the Hawking mass blows up when one approaches $\CH$ over (at least) \textbf{one} outgoing light cone. From these two facts, we show a contradiction. As a consequence of our analysis, we prove that a non-trivial boundary component must emanate from the center, with two possibilities (see already the upcoming Theorem~\ref{roughversion}): \begin{enumerate} 	\item either a $r=0$ type singularity $\mathcal{S}$ is present, in addition to the Cauchy horizon, see Figure \ref{Fig2},
 		\item or a null outgoing segment emanates from the endpoint of the center $b_{\Gamma}$ to meet the Cauchy horizon, see Figure~\ref{Fig1} 	\end{enumerate}
 	\begin{figure}[H]
 	
 	\begin{center}
 		
 		\includegraphics[width=77 mm, height=47 mm]{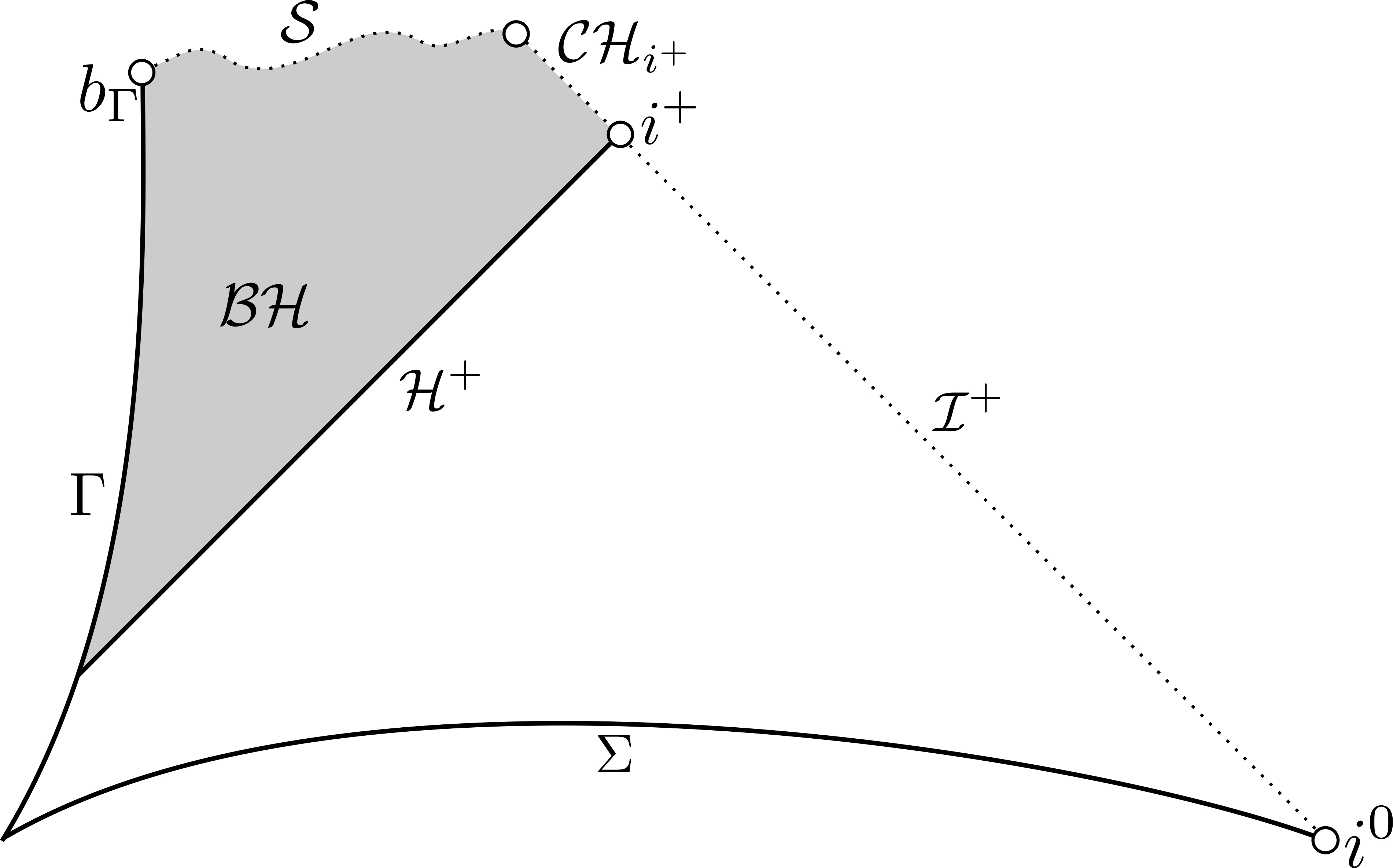}
 		
 	\end{center}
 	
 	\caption{{Generic} Penrose diagram of a one-ended charged black hole under the assumptions of Theorem \ref{spaceliketheorem}.}
 	\label{Fig2}
 \end{figure}
 	
 	The second possibility, which corresponds to a ``locally naked'' singularity emanating from the center, is conjectured to be non-generic. In both cases, we prove that the endpoint of the center $b_{\Gamma}$ is a (central) first singularity. In the first case, the boundary additionally contains a (non-central) first singularity belonging to $\mathcal{S}$, see section \ref{spacelikeconjsection}.

 	We also obtain a second result if, instead of assuming mass inflation, we make assumptions on the event horizon, see the upcoming Theorem~\ref{rough2}. More precisely, we assume that the scalar field decays on the event horizon at a weak polynomial rate which is conjectured to be generic. Then we show, using the main result of \cite{Moi4}, that  \begin{itemize} \item either the Cauchy horizon features mass inflation, hence it cannot close off the space-time by our new result,
 		\item or the Cauchy horizon is an isometric copy of the Reissner--Nordstr\"{o}m Cauchy horizon.  \end{itemize}
 	In the latter case, we prove that the Cauchy horizon cannot close off the space-time \emph{either}, using a different rigidity-type argument. Our assumptions are comparable to those used in \cite{Moi} to prove the non-emptiness of the Cauchy horizon and backed up by multiple numerical studies, see section \ref{hypsection}. Moreover, the decay assumptions that we use are quite weak since we require much less information on the outgoing data than what is conjectured to be generic.

 	We now return to the main result [Theorem~\ref{roughversion}], assuming mass inflation. Our proof relies on quantitative estimates, in which the center of symmetry $\Gamma$ plays a role of utmost importance (recall that it is the presence of a regular center $\Gamma$ that distinguishes gravitational collapse space-times from the two-ended case considered in \cite{Mihalisnospacelike}). The most significant quantity that we control is the Maxwell field, which is dynamical as it interacts with the charged scalar field. As we explained earlier, if this interaction was not present in the equations -- like for the Dafermos model -- no regular one-ended solution could exist, as the center of symmetry $\Gamma$ would be singular, which is impossible. Thus, we emphasize that the quantitative estimates of this Maxwell field, in particular \underline{near the center}, are a fundamental aspect of the problem, which cannot be overlooked by any serious  attempts to establish the generic character of $r=0$ singularities.
 	
 	In section \ref{apriori}, we give a detailed description of the matter model, the Einstein--Maxwell--Klein--Gordon equations and we enumerate all the possible a priori Penrose diagrams, following \cite{Kommemi}. Then, we state our main result and discuss its assumptions in section \ref{hypsection}. We mention previous works on uncharged models in section \ref{ChrisDaf}. In section \ref{spacelikeconjsection}, we explain the $r=0$ singularity conjecture and its relation to first singularities. We continue with \color{black}  a presentation of other major problems, solved or unsolved, in charged collapse in section \ref{additionnal}. The few heuristic and numerical previous works on one-ended black holes are mentioned in section \ref{numerical}. In section \ref{twoended}, we mention prior results on two-ended black holes \color{black} and emphasize the contrast with the one-ended case. Finally in section \ref{outline}, we give a short outline of the proof.

 	\subsection{A priori boundary characterization of one-ended space-times} \label{apriori}
 	
 	We consider the Einstein--Maxwell--Klein--Gordon equations, namely the Einstein equations in the presence of a charged scalar field (which we also allow to be massive if $m^2 \neq 0$ or massless if $m^2=0$) i.e.\ \eqref{1.1}, \eqref{2.1}, \eqref{3.1}, \eqref{4.1}, \eqref{5.1}, where  $D:= \nabla+  iq_{0}A$ is the gauge derivative, $\nabla$ is the Levi-Civita connection of $g$ and $A$ is the potential one-form. We emphasize that the Klein--Gordon mass $m^2 \geq 0$ is allowed to be zero, but not the coupling constant $q_0 \neq 0$. 
 	
 	Some a priori information can be derived for this system in spherical symmetry, from ``soft estimates'' only involving the null condition satisfied by the non-linearity. This work was carried out by Kommemi in \cite{Kommemi}, who gave an inventory of the vast a priori possibilities for the interior structure of the black hole. To determine which boundary components are empty or singular, one must go beyond such ``a priori estimates`'', and a precise analysis of the equations is required, which we undertake in the present work. Note that in the sequel, \textit{Penrose diagrams} themselves will be used to convey important information regarding the spacetime geometry, see Figure~\ref{Fig2},  Figure~\ref{Figure5}, Figure~\ref{Figrect} and see the discussion of spherically symmetric spacetimes in Section~\ref{preliminary} and in \cite{Kommemi}. We now present the preliminary result of Kommemi.

 	\begin{theotwo}[Theorem 1.1\ of \cite{Kommemi}\color{black}] \label{Kommemi}
 		
 		We consider the maximal development $(M=\mathcal{Q}^+ \times_r \mathcal{S}^2,g_{\mu \nu}, \phi,F_{\mu \nu})$ of smooth, spherically symmetric, containing no anti-trapped surface, one-ended asymptotically flat \color{black} initial data satisfying the Einstein--Maxwell--Klein--Gordon system, where $r: \mathcal{Q}^+ \rightarrow [0,+\infty)$ is the area-radius function. Then the Penrose diagram of $\mathcal{Q}^+$ is given by Figure \ref{Fig1}, with boundary $\Sigma \cup \Gamma$ in the sense of manifold-with-boundary --- where $\Sigma$ is space-like, and $\Gamma$, the center of symmetry, is time-like with $r_{|\Gamma}=0$ --- and boundary $\mathcal{B}^+$ induced by the manifold ambient $\RR^{1+1}$: $$ \mathcal{B}^+ = b_{\Gamma} \cup \mathcal{S}^1_{\Gamma} \cup \mathcal{CH}_{\Gamma} \cup  \mathcal{S}^2_{\Gamma} \cup  \mathcal{S}  \cup \mathcal{S}_{i^+}  \cup \CH \cup i^{+} \cup \mathcal{I}^+ \cup i^0,$$ where $i^0$ is space-like infinity, $\mathcal{I}^+$ is null infinity, $i^{+}$ is time-like infinity (see \cite{Kommemi} for details) and 	\begin{figure}
 			
 			\begin{center}
 				
 				\includegraphics[width=83.5 mm, height=60 mm]{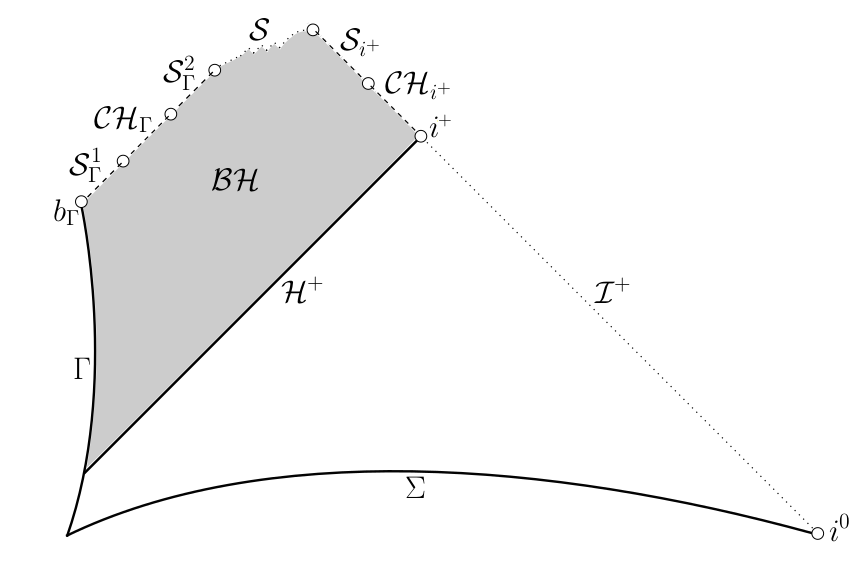}
 				
 			\end{center}
 			
 			\caption{General Penrose diagram of a one-ended charged spherically symmetric black hole, \cite{Kommemi}.} 
 			
 			\label{Fig1}
 		\end{figure}  \begin{enumerate}
 		\item $\CH$ is a connected (possibly empty) half-open null ingoing segment emanating from $i^{+}$. The area-radius function $r$ extends as a strictly positive function on $\CH$, except maybe at its future endpoint.
 		\item $\mathcal{S}_{i^+}$ is a connected (possibly empty) half-open null ingoing segment emanating (but not including) from the end-point of $\CH \cup i^{+}$. $r$ extends continuously to zero on $\mathcal{S}_{i^+}$.
 		\item $b_{\Gamma}$ is the center end-point i.e.\ the unique future limit point of $\Gamma$ in $\overline{\mathcal{Q}^+}-\mathcal{Q}^+$.
 		\item $ \mathcal{S}^1_{\Gamma} $ is a connected (possibly empty) half-open null outgoing segment emanating from $b_{\Gamma}$.\\ $r$ extends continuously to zero on $ \mathcal{S}^1_{\Gamma} $.
 		\item $\mathcal{CH}_{\Gamma}$ is a connected (possibly empty) half-open null outgoing segment emanating from the future end-point of $b_{\Gamma} \cup  \mathcal{S}^1_{\Gamma} $. $r$ extends as a strictly positive function on $\mathcal{CH}_{\Gamma}$, except maybe at its future endpoint.
 		\item $\mathcal{S}^2_{\Gamma}$ is a connected (possibly empty) half-open null outgoing segment emanating from the future end-point of $\mathcal{CH}_{\Gamma}$.  $r$ extends continuously to zero on $\mathcal{S}^2_{\Gamma}$. 	\item $\mathcal{S} $ is a connected (possibly empty) achronal curve that does not intersect null rays emanating from $b_{\Gamma}$ or $i^+$. $r$ extends continuously to zero on $\mathcal{S}$.
 	\end{enumerate} 
 	We also define the black hole region $ \mathcal{BH}:= \mathcal{Q}^+ \backslash J^{-}(\mathcal{I}^+)$, and the event horizon $\mathcal{H}^+ = \overline{J^{-}(\mathcal{I}^+)} \backslash J^{-}(\mathcal{I}^+) \subset~ \mathcal{Q}^+$.
 \end{theotwo}

 \begin{rmk}
 	$\mathcal{S}$ is the only boundary component including ``first singularities'', see section \ref{firstsing} for a discussion.
 \end{rmk}

 In the Penrose diagram, every point in $M$ represents a sphere. At each sphere, one can define the outgoing null derivative of the area-radius function $r$. We define the regular region, denoted $\R$ as the set of points for which the outgoing null derivative of $r$ is strictly positive, the trapped region, denoted $\T$ as the set of points for which the outgoing null derivative of $r$ is strictly negative and the apparent horizon, denoted $\A$ as the set of points for which the outgoing null derivative of $r$ is zero\footnote{The definitions \color{black} of $\R$, $\T$, and $\A$ rely \color{black} on a sign, and thus is independent of the null vectors normalization, see Section~\ref{geometricframework}.}. 
 
 An important feature of charged matter models like the Einstein--Maxwell--Klein--Gordon system is that they permit spacetimes with a regular center $\Gamma\subset \{r=0\}$ [corresponding to constant-time slices \color{black} with one asymptotically flat end]. Indeed, in spherical symmetry, the Maxwell field can be written in terms of a scalar function $Q$ defined on $\mathcal{Q}^+$ as: $$ F_{\mu \nu} = \frac{Q}{r^2} \cdot \Omega^2 du \wedge dv,$$ in any double null coordinate system $(u,v)$ on $\mathcal{Q}^+$, where we defined the null lapse to be $\Omega^2 =-2 g(\partial_u,\partial_v)$.

 In the case $q_0=0$, the right-hand-side of \eqref{4.1} vanishes and this implies that $Q \equiv e$ is a constant function. If $e \neq 0$ then $F_{\mu \nu}$ diverges at the centre $\Gamma=\{r=0\}$: thus, no one-ended smooth solution is possible in the uncharged matter case, e.g. in the Dafermos model. In the setting of charged matter, $Q$ is no longer a constant function and one-ended regular solutions are available, providing we impose boundary conditions at the center $\Gamma$, detailed in section \ref{geometricframework}.

 \subsection{First version of the main results, and discussion of the hypothesis} \label{hypsection}
 In this section, we state our main result, namely that a weakly singular Cauchy horizon cannot close off the space-time. For the consequences on the $r=0$ singularity conjecture and the generic existence of first singularities, see section \ref{spacelikeconjsection}.
 
 \subsubsection{Theorem assuming the existence of a weak null singularity}
 
 We present our main theorem, which does not require any quantitative assumption. Namely, \textit{we only assume the existence of a weak null singularity}, i.e.\ the blow up of the mass and the boundedness of the matter fields over one outgoing trapped cone reaching the Cauchy horizon. We give a first version of the theorem, which will be made precise in Section~\ref{statement}.
 \begin{theo} \label{roughversion}
 	For initial data as in Theorem \ref{Kommemi}, assume there exists a trapped cone $\{u_1\} \times [v_1,v_{max})$ with $(u_1,v_{max}) \in \CH$, on which the Hawking mass $\rho$ blows up, while $\phi$ and $Q$ are bounded: \begin{equation} \label{roughassumption}
 	\lim_{v \rightarrow v_{max}}	\rho(u_1,v)= +\infty, \hskip 5 mm 	\sup_{v \leq v_{max}}	|\phi|(u_1,v)+ |Q|(u_1,v)< +\infty .\end{equation} 	
 	Then $$\mathcal{S}^1_{\Gamma} \cup \mathcal{CH}_{\Gamma} \cup  \mathcal{S}^2_{\Gamma} \cup  \mathcal{S}  \neq \emptyset,$$ i.e.\ $\CH \cup \mathcal{S}_{i^+}$ cannot close off the space-time at $b_{\Gamma}$, i.e.\ the Penrose diagram of Figure \ref{Fig3} is impossible.

 \end{theo}
 
 The theorem relies on two assumptions: the first one is the blow-up of the Hawking mass towards (at least) \textbf{one} sphere on the Cauchy horizon. This blow-up, which is conjectured to be generic, results from the blue-shift effect of radiation at the Cauchy horizon. This effect is localized near $i^+$ inside the black hole, in the sense that no knowledge of the global structure of space-time is necessary to obtain it, as it only depends on the asymptotic structure of the event horizon, see \cite{MihalisSStrapped}, \cite{Moi4}. Of course, the mass blow up of \eqref{roughassumption} is not satisfied for the Reissner--Nordstr\"{o}m solution: this is simply due to the absence of radiation in the Reissner--Nordstr\"{o}m space-time, which is static. The other assumption is the boundedness of $\phi$ and $Q$, again over the same outgoing light cone, which is expected to hold generically as well, due to local stability estimates near time-like infinity, see \cite{Moi} and \cite{Moi3Christoph}.
 
 \subsubsection{Theorem with assumptions on the event horizon}
 
 Instead of making assumptions on the behavior of one given light cone in the black hole interior, one can also make assumptions on the event horizon $\mathcal{H}^+$ to prove that the Cauchy horizon $\CH$ does not close the space-time. We give a first version of the theorem, which will be made precise later in Section~\ref{statement}. \begin{theo} \label{rough2}
 	We assume that the black hole event horizon $\mathcal{H}^+$ settles quantitatively towards a sub-extremal Reissner--Nordstr\"{o}m event horizon. Then $\mathcal{S}^1_{\Gamma} \cup \mathcal{CH}_{\Gamma} \cup  \mathcal{S}^2_{\Gamma} \cup  \mathcal{S}  \neq \emptyset.$
 \end{theo}
 
 This theorem uses the following result of the author: assuming the quantitative exterior stability, we proved in \cite{Moi4} \begin{enumerate}
 	\item \label{roughtalt1} that either the Hawking mass blows up on the Cauchy horizon, while $\phi$ and $Q$ are controlled\footnote{We cannot work under the assumption  that $(\phi,Q)$ are bounded, since this is not true in general, see \cite{Moi3Christoph}. However, we have proved in \cite{Moi4}  that $(\phi,Q)$ do not blow-up too fast, which is what we mean by ``controlled''} 
 	\item \label{roughtalt2} Or the Cauchy horizon $\CH$ is an isometric copy of the Reissner--Nordstr\"{o}m one.
 \end{enumerate}
 
 If the first option is true, then the assumptions of Theorem \ref{roughversion} are satisfied which implies that the Cauchy horizon cannot close the space-time: $\mathcal{S}^1_{\Gamma} \cup \mathcal{CH}_{\Gamma} \cup  \mathcal{S}^2_{\Gamma} \cup  \mathcal{S}  \neq \emptyset$. If the second option is true, then, we prove that it is impossible to connect an isometric copy of the Reissner--Nordstr\"{o}m Cauchy horizon to a space-time for which $\mathcal{S}^1_{\Gamma} \cup \mathcal{CH}_{\Gamma} \cup  \mathcal{S}^2_{\Gamma} \cup  \mathcal{S}  = \emptyset$. To do this, we make use of an argument showing that focusing cannot occur in the vicinity of the Cauchy horizon, together with geometric properties, see section \ref{outline}.
 

 \subsubsection{Conjectured decay rates on the event horizon} \label{sectionrates}
 In this sub-section, we discuss the conjectured decay rates at which a black hole is expected to settle towards a sub-extremal Reissner--Nordstr\"{o}m black hole for large times, and the previous  works on the subject in the Physics literature. In our upcoming Theorem~\ref{rough2}, we will assume these rates are satisfied on the event horizon of the dynamical black hole.
 
 In \cite{HodPiran1}, Hod and Piran provided a heuristic argument, based on asymptotic matching, to formulate the correct decay of charged scalar fields on charged spherically symmetric black holes. The main difference with uncharged fields is that the decay rate now depends on the black hole charge $e$, as opposed to the universal rate prescribed by Price's law \cite{PriceLaw,Schlag,JonathanStabExt,Tataru} in the uncharged case. The results of \cite{HodPiran1} are confirmed by the numerics of Oren and Piran \cite{OrenPiran}, providing further evidence that the rate depends on $q_0 e$, the adimensional black hole charge.
 \begin{conjecture} [Decay of charged scalar fields, Hod and Piran \cite{HodPiran1}, Oren and Piran \cite{OrenPiran}] \label{chargedconj} Among all the data admissible and sufficiently regular and decaying from Theorem \ref{Kommemi}, there exists a generic sub-class for which if the maximal future development has $\mathcal{Q}^+ \cap J^{-}(\mathcal{I}^+)\neq \emptyset$, we have, in the charged massless case $q_0 \neq 0$, $m^2=0$:
 	$$ |\phi|_{|\mathcal{H}^+}(v) \sim v^{-2+\delta(q_0e)}, \hskip 5 mm |D_v \phi|_{|\mathcal{H}^+}(v) \sim v^{-2+\delta(q_0e)},$$ where $e$ is asymptotic charge of the black hole at time-like infinity, $\delta(q_0e):=1- \Re(\sqrt{1-4(q_0e)^2}) \in [0,1)$ and $v$ is a null coordinate defined by the gauge choice \eqref{gauge1}.
 \end{conjecture}
 
 Note that the upper bound corresponding to conjecture \ref{chargedconj} was retrieved rigorously in \cite{Moi2}, on a fixed Reissner--Nordstr\"{o}m background, for small charge $q_0e$ and for a rate $p= 2-\delta(q_0e)+o( \sqrt{|q_0e|})$ as $q_0e \rightarrow 0$.
 
 Now we turn to the case of a massive scalar field. In \cite{KoyamaTomimatsu}, Koyama and Tomimatsu considered the case of a massive, uncharged scalar field and provided a heuristic argument, also based on asymptotic matching, to support that massive fields decay polynomially, at a very weak rate and with oscillations. Their tails were later confirmed by the numerics of Burko and Khanna \cite{BurkoKhanna}. For the case of a massive, charged scalar field, it was argued by Konoplya and Zhidenko \cite{KonoplyaZhidenko} that the late-time tail must be identical, as they claim that the asymptotic behavior of massive scalar field is universal.
 
 \begin{conjecture} [Decay of massive scalar fields, \cite{BurkoKhanna,KoyamaTomimatsu,KonoplyaZhidenko}] \label{conjecturemassive} Among all the data admissible, sufficiently regular and decaying from Theorem \ref{Kommemi}, there exists a generic sub-class for which if the maximal future development has $\mathcal{Q}^+ \cap J^{-}(\mathcal{I}^+)\neq \emptyset$, we have, in the massive uncharged case $m^2 \neq 0$, $q_0=0$ \cite{BurkoKhanna,KoyamaTomimatsu}, and in the massive charged case $m^2 \neq 0$, $q_0 \neq 0$ \cite{KonoplyaZhidenko}:
 	$$ |\phi|_{|\mathcal{H}^+}(v) \sim |\sin|( mv + o(v) )\cdot v^{-\frac{5}{6}}, \hskip 5 mm |D_v \phi|_{|\mathcal{H}^+}(v) \sim |\sin|( mv + o(v) )\cdot v^{-\frac{5}{6}},$$ where $v$ is a null coordinate defined by the gauge choice \eqref{gauge1}.
 \end{conjecture}

 \subsubsection{Comments on the assumptions of Theorem \ref{roughversion} and Theorem \ref{rough2}}
 
 The assumptions we make for Theorem \ref{rough2} (see Theorem \ref{theoremevent} for the details) are compatible with the tails of Conjecture \ref{chargedconj} and Conjecture \ref{conjecturemassive} but require much less information: we only assume, for  some $s>\frac{3}{4}$, \begin{equation*} 
 |\phi|_{|\mathcal{H}^+}(v) + |D_v \phi|_{|\mathcal{H}^+}(v) \lesssim v^{-s},	\end{equation*} \begin{equation*} 
 \int_{v}^{+\infty} |D_v \phi|^2_{|\mathcal{H}^+}(v')dv' \gtrsim v^{1-2s}.	\end{equation*}
 
 \begin{rmks}
 Notice that $1>\frac{5}{6}>\frac{3}{4}$ so this decay is compatible with both Conjecture  \ref{chargedconj} and Conjecture \ref{conjecturemassive}.
 \end{rmks}
 Note that we require an $L^2$-averaged (energy) polynomial lower bound (which is weaker than a point-wise bound) to also account for the potential oscillations prescribed by Conjecture \ref{conjecturemassive}.
 
 
 As for Theorem \ref{roughversion}, assumption \eqref{roughassumption} can be weakened, if we allow global, but soft assumptions. More precisely, the mildest assumption that will suffice for our result, instead of the mass blow up, is some integrability condition \eqref{mainassumption} over (only) \textbf{one} outgoing cone, in addition to the assumption that the entire Cauchy horizon is trapped, see Theorem \ref{maintheorem} for a precise statement. In turn, we prove that this scenario occurs if one only assumes mass blow up on one cone, due to the propagation of blow up, see section \ref{propagationsection}. Thus, we emphasize that the blow up of the Hawking mass is just a \underline{sufficient} condition but it is not necessary 
 to prove that the Cauchy horizon cannot close the space-time. In the asymptotically flat setting, mass inflation occurs \cite{Ori,Moi4} so the sufficiency of these weaker assumptions does not matter. However, in other settings such as the Einstein equations with a positive cosmological constant, mass inflation is not generically expected \cite{Costa,Harvey}, but the weaker assumption may be still satisfied; thus we expect our argument to be useful in this  setting as well.  
 \subsection{The models of Christodoulou and Dafermos and their generalization} \label{ChrisDaf}
 
 In this section, we present two sub-models of the Einstein--Maxwell--Klein--Gordon equations. The first one is the uncharged spherically symmetric model studied by Christodoulou \cite{Christo1,Christo2,Christonaked,Christo3}, governed by the Einstein-scalar-field equations, i.e.\ the system  \eqref{1.1}, \eqref{2.1}, \eqref{3.1}, \eqref{4.1}, \eqref{5.1} in the special case $F \equiv 0$, $m^2=0$. While this model is suitable to study gravitational collapse, as one-ended solutions are allowed, it does not permit the formation of Cauchy horizons, due to the absence of any repulsive mechanism such as angular momentum or charge. Therefore, this model is ill-suited to understand weak null singularities, as there is no Cauchy horizon emanating from time-like infinity $i^+$. 
 
 The second model, featuring a Maxwell field with \underline{uncharged} matter was studied by Dafermos \cite{MihalisPHD,Mihalis1,MihalisSStrapped,Mihalisnospacelike}, and is governed by the Einstein--Maxwell-(uncharged)-scalar-field equations, i.e.\ the system  \eqref{1.1}, \eqref{2.1}, \eqref{3.1}, \eqref{4.1}, \eqref{5.1} in the special case $q_0=0$, $m^2=0$. This model allows Cauchy horizons to form and provides a good setting to understand the formation of weak null singularities and their \textit{local} aspects. Yet, the Dafermos model is in turn restricted by the topology of its initial data, necessarily \textbf{two}-ended. This is because the Maxwell field, which is static due to the absence of charged matter, is singular in the one-ended case. Therefore, the Dafermos model is inappropriate to study the \textit{global} aspects of gravitational collapse, including the necessary breakdown of weak null singularities, due to the absence of a center $\Gamma$.

 
 The Einstein--Maxwell--Klein--Gordon equations in spherical symmetry that we study in the present paper generalize both the Christodoulou and the Dafermos model, and are free from the above restrictions, as one-ended black holes with Cauchy horizon are allowed in principle. In fact, the Einstein--Maxwell--Klein--Gordon system is one\footnote{It is also possible to study charged dust, but dust on its own propagates by transport. The dust model is simpler than the  Einstein--Maxwell--Klein--Gordon system, but  is thus not expected to reflect as accurately the dynamics of the vacuum Einstein equations \cite{violation}.\color{black}} of the  only \color{black} spherically symmetric model which is elaborate enough to formulate the breakdown of weak null singularities in a non-trivial way. 
 
 An important preliminary step, before proving the result of our present paper, is to establish that the Cauchy horizon is indeed \underline{always} non-empty\footnote{Providing, of course, that the asymptotic charge $e$ of the black hole is non-zero, which is conjecturally the generic case.} in the Einstein--Maxwell--Klein--Gordon model \cite{Moi}. As the dynamics of charged scalar fields in the exterior are more intricate than their uncharged counterparts, new difficulties arise. These difficulties were  overcome by the author in \cite{Moi,Moi4}, where it was also shown that the Cauchy horizon $\CH$ is weakly singular.
 
 We now briefly present this result, after mentioning previous works for context.
 \subsubsection{The Einstein-scalar-field equations in spherical symmetry}
 The uncharged gravitational collapse has been analysed in great detail by Christodoulou. Recall that in this case, no Cauchy horizon is allowed to emanate from time-like infinity $i^+$ due to the absence of charge or angular momentum. We sum up Christodoulou's main results on the black hole interior.
 \begin{thm}[Christodoulou, Einstein-scalar-field in spherical symmetry \cite{Christo1,Christo2,Christo3}] \label{Christotheorem}
 	For initial data as in Theorem \ref{Kommemi} in the more general BV class, assume that the Maxwell field is trivial: $F_{\mu \nu} \equiv 0$ and that the field is massless $m^2=0$. Then: \begin{enumerate}
 		\item There is no Cauchy horizon emanating from time-like infinity: $\CH = \emptyset$.
 		
 		\item There is no secondary outgoing null segment emanating from $b_{\Gamma}$ and where $r=0$: $\mathcal{S}_{\Gamma}^{2} =\emptyset$.
 		
 		\item \label{trappedsurfaceconsequence} Among all the data admissible from Theorem \ref{Kommemi}, in the BV class, with $F_{\mu \nu} \equiv 0$, $m^2=0$, there exists a generic sub-class for which if the maximal future development has $\mathcal{Q}^+ \cap J^{-}(\mathcal{I}^+)\neq \emptyset$, then $\mathcal{S}$ is the only non-trivial component of the boundary: $\mathcal{S}_{\Gamma}^{1}= \mathcal{CH}_{\Gamma} =\emptyset$.
 	\end{enumerate}
 \end{thm}
 Notice that the statement  $\mathcal{S} \neq \emptyset$ is immediate for the Christodoulou model, where very special monotonicity properties dominate, in the absence of any repulsive mechanism such as angular momentum or charge. This is in contrast with the model considered in the present paper, where the (non-trivial) presence of a Cauchy horizon $\CH$ (see section \ref{mywork}) could, in principle, allow for $\CH$ to be the only non-empty boundary component.
 \subsubsection{The Einstein--Maxwell-scalar-field equations in spherical symmetry}
 
 The breakthrough of Dafermos \cite{MihalisPHD} was to realize and prove that the Cauchy horizon is non-empty for dynamical black holes, in a model where a Maxwell field plays the role of angular momentum. This early insight, gained from a spherically symmetric model, paved the way to the monumental work of Dafermos and Luk \cite{KerrStab} who recently proved the stability of the Cauchy horizon of Kerr black holes for the vacuum Einstein equations, remarkably in the absence of any symmetry.
 \begin{thm}[Dafermos, Einstein--Maxwell-scalar-field in spherical symmetry \cite{MihalisPHD,Mihalis1}] 
 	Assume that the black hole event horizon $\mathcal{H}^+$ settles quantitatively towards a sub-extremal Reissner--Nordstr\"{o}m event horizon. Then $\CH \neq \emptyset$, i.e.\ there exists a non-empty Cauchy horizon emanating from $i^+$. Moreover, $\CH$ is weakly singular, and yet $C^0$ extendible.
 \end{thm}
 The $C^0$ extendibility of $\CH$, together with a result on the exterior by Dafermos and Rodnianski \cite{PriceLaw}, also falsifies the $C^0$ version of Strong Cosmic Censorship  in spherical symmetry\color{black}, see section \ref{additionnal} for a discussion of this important conjecture related to determinism.
 
 We emphasize, however, that the theorem of Dafermos is for uncharged scalar fields, and thus does not apply to the Einstein--Maxwell--Klein--Gordon model -- featuring a \emph{charged}, massive scalar field -- considered in the present paper.
 \subsubsection{Non-emptiness of $\CH$ for the Einstein--Maxwell--Klein--Gordon model} \label{mywork}
 In the charged and massive case, the model becomes more intricate and the proof of Dafermos does not carry over. Additionally, the scalar field obeys different dynamics, in particular a weaker decay than in the uncharged case, see Conjecture \ref{chargedconj} and Conjecture \ref{conjecturemassive} and compare with the Price's law governing uncharged fields \cite{PriceLaw}. This weak decay of charged/massive fields renders non-linear stability harder and requires new estimates, established by the author in \cite{Moi}:
 \begin{thm}[Einstein--Maxwell--Klein--Gordon in spherical symmetry \cite{Moi}]  
 	Assume that the event horizon $\mathcal{H}^+$ settles quantitatively towards a sub-extremal Reissner--Nordstr\"{o}m event horizon. Then $\CH \neq \emptyset$, and $\CH$ is weakly singular. 
 \end{thm} \begin{rmk}
 In this context, the weak null singularity of $\CH$ is to be understood as a blow up of a curvature component.
\end{rmk} Under the same assumptions, it is also proven in \cite{Moi} that $\CH$ is $C^0$ extendible, in the charged massless case. Therefore, the $C^0$ version of Strong Cosmic Censorship (see section \ref{additionnal}) is also false in this more general setting. The same conclusion was reached for the massive and charged case by the author and Kehle \cite{Moi3Christoph}, exploiting a novel mechanism based on the oscillations of the perturbations.

\subsection{The $r=0$ singularity conjecture in charged gravitational collapse} \label{spacelikeconjsection}

Theorem \ref{roughversion} has consequences on space-time singularities: under some reasonable additional assumptions, one can prove that a $r=0$ singularity exists generically in the black hole interior. We discuss this question in the present section.

\subsubsection{First singularities}
We define the notion of spherically symmetric ``first singularities'', a concept introduced by Dafermos in \cite{MihalisSStrapped} and further formalized by Kommemi in \cite{Kommemi}. First singularities are the boundary points from which non-trivial components emanate. Thus, most of the investigation of the black hole interior relies on the precise understanding of those singularities.
\begin{defunb}
	Let  $p \in \overline{\mathcal{Q}^+}$: we say that $J^{-}(p) \subset \overline{\mathcal{Q}^+} $ is compactly generated if there exists a compact set $X \subset \mathcal{Q}^+$ such that $ J^{-}(p) \subset \Pi( \mathcal{D}^+_M(\Pi^{-1}(X)))  \cup J^{-}(X),$ where $\mathcal{D}^+_M(A)$ is the domain of dependence of $A$, and $\Pi: M \rightarrow \mathcal{Q}^+$ is the projection onto the Penrose diagram $ \mathcal{Q}^+$ (see the definition in Section~\ref{geometricframework}).
\end{defunb}
\begin{defn} [\cite{Kommemi}] \label{firstsing}
	With the conventions of Theorem \ref{Kommemi}, we say that $p \in \mathcal{B}^+$ is a first singularity if $J^-(p)$ is compactly generated  and if any compactly generated proper causal subset of $J^-(p)$ is of the form $J^-(q)$, $q \in \mathcal{Q}^+$. Then \begin{enumerate}
		\item If $p \in \mathcal{B}^+ -\{b_{\Gamma}\}$, $p$ is a first singularity if and only if there exists $q \in I^{-}(p) \cap \mathcal{Q}^+ \backslash  \{p\}$, such that $J^+(q) \cap J^{-}(p) \backslash \{p\} \subset  \mathcal{Q}^+$: we say that $p$ is a non-central first singularity,
		\item $b_{\Gamma}$ is a first singularity if and only if $\mathcal{S} \cup \mathcal{S}_{\Gamma}^{1} \cup \mathcal{CH}_{\Gamma} \cup\mathcal{S}_{\Gamma}^{2} \neq \emptyset$. In that case, we say $b_\Gamma$ is a central first singularity.
	\end{enumerate}

\end{defn}
Theorem \ref{Kommemi} then has an immediate application on the location of the first singularities:
\begin{corunb}[Corollary of Theorem \ref{Kommemi}, \cite{Kommemi}]
	If $\mathcal{S} \neq \emptyset$, there exists at least one non-central first singularity $p\in \mathcal{S}$ and $b_{\Gamma}$ is a central first singularity.
	Moreover, there are \textbf{no} non-central first singularities in $\mathcal{B}^+ - \mathcal{S}$: therefore, if $p$ is a non-central first singularity, then $r(p)=0$.
\end{corunb}
In view of Theorem \ref{Kommemi}, one can also define a reasonable notion of space-like portion of the boundary: \begin{defn} \label{spacelikedef}
	Let $\mathcal{S}' \subset \mathcal{S}$. We say that $\mathcal{S}'$ is space-like if every $p\in \mathcal{S}'$ is a first singularity.
	
\end{defn}
\begin{rmk}
	Notice, even if $\mathcal{S}$ is space-like in the sense of Definition \ref{spacelikedef}, it is not clear a priori whether one can, from $\mathcal{S}$, construct and attach a space-like boundary to the $3+1$ space-time $\mathcal{Q}^+ \times_r \mathbb{S}^2$, as subtle considerations may be important.
\end{rmk}

For the uncharged model of Christodoulou, i.e.\ \eqref{1.1}, \eqref{2.1}, \eqref{3.1}, \eqref{4.1}, \eqref{5.1} in the special case $F_{\mu \nu } \equiv 0$, $m^2=0$ the analysis of Christodoulou \cite{Christo1,Christo2,Christo3} leading to Theorem \ref{Christotheorem} also impacts the structure of first singularities: \begin{propunb}[First singularities for Einstein-scalar-field, Christodoulou \cite{Christo1}, \cite{Christo2}, \cite{Christo3}, Dafermos \cite{MihalisSStrapped}]
	Among all the data admissible from Theorem \ref{Christotheorem}, there exists a generic sub-class for which if the maximal future development has $\mathcal{Q}^+ \cap J^{-}(\mathcal{I}^+)\neq \emptyset$, then every point $p\in \mathcal{B}^+$ is a first singularity. In particular, $\mathcal{S}$ is space-like.
	
\end{propunb}

\subsubsection{The $r=0$ singularity conjecture}

Now, we return to the charged gravitational collapse case, i.e.\ the full system \eqref{1.1}, \eqref{2.1}, \eqref{3.1}, \eqref{4.1}, \eqref{5.1} where $F_{\mu \nu} \neq 0$, $q_0 \neq 0$ and we mention important conjectures formulated in \cite{Kommemi}. In view of Definition \ref{firstsing}, our main result directly implies: \begin{cor}
	Under the assumptions of Theorem \ref{roughversion} or Theorem \ref{rough2}, $b_{\Gamma}$ is a central first singularity, therefore the set of first singularities is non-empty.
\end{cor}
Now, we want to investigate the behavior of the area-radius $r$ at the interior boundary $\mathcal{B}^+$. We introduce an important conjecture stating the existence of a crushing singularity $\mathcal{S}=\{r=0\}$.		\begin{conjecture}[$r=0$ singularity conjecture, as formulated in \cite{Kommemi}] \label{spacelikeconj}
	Among all the data admissible from Theorem \ref{Kommemi}, there exists a generic sub-class for which if the maximal future development has $\mathcal{Q}^+ \cap J^{-}(\mathcal{I}^+)\neq \emptyset$, then the Penrose diagram is given by Figure \ref{Fig2} i.e.\ $\mathcal{S} \neq \emptyset$, $\CH \neq \emptyset$ and $\mathcal{S}_{\Gamma}^{1}= \mathcal{CH}_{\Gamma} =\mathcal{S}_{\Gamma}^{2} =\emptyset$.
\end{conjecture} The main content of the conjecture is the statement $\mathcal{S} \neq \emptyset$, which implies that there exists a non-trivial boundary component where $r=0$. Additionally, $\mathcal{S} \neq \emptyset$ implies the existence of a (non-central) first singularity, by Definition \ref{firstsing}. The additional assumption we need to prove Conjecture \ref{spacelikeconj} in the present paper is formulated as another conjecture, which has important connections with the Weak Cosmic Censorship conjecture (see section \ref{additionnal}):
\begin{conjecture}[Spherical trapped surface conjecture, formulated in \cite{Kommemi}] \label{trappedsurfaceconj}
	Among all the data admissible from Theorem \ref{Kommemi}, there exists a generic sub-class for which, if the maximal future development has $\mathcal{Q}^+  -\color{black} J^{-}(\mathcal{I}^+) \neq \emptyset$, then the apparent horizon $\A$ has a limit point on $b_{\Gamma}$. Moreover, if that is the case then $\mathcal{S}_{\Gamma}^{1}= \mathcal{CH}_{\Gamma} =\mathcal{S}_{\Gamma}^{2}=\emptyset$.
\end{conjecture} \begin{rmk}
Note that the statement $\mathcal{S}_{\Gamma}^{1}= \mathcal{CH}_{\Gamma} =\mathcal{S}_{\Gamma}^{2}=\emptyset$ corresponds to the absence of ``locally naked'' singularity emanating from the center $b_{\Gamma}$, a slightly stronger statement than Weak Cosmic Censorship (Conjecture \ref{WCC}).
\end{rmk}

Notice that Conjecture \ref{trappedsurfaceconj} is related to the behavior of space-time in the vicinity of $b_{\Gamma}$, therefore, by causality, this behavior cannot be influenced by the late time tail on the event horizon\footnote{Except, of course, in the case where $\CH$ closes off the spacetime as in Figure~\ref{Fig3}, but it is precisely the case we disprove in Theorem~\ref{roughversion}}. In contrast, our results start from data on a fixed outgoing trapped cone, itself related to asymptotic behavior on the event horizon ultimately responsible for the existence of a weak null singularity. In fact, by the same principle, the result of \cite{Moi} and Theorem \ref{roughversion} are the only emptiness/non-emptiness statements which can be non trivially obtained from the late time behavior on an outgoing cone, as all the other possible statements would result from purely local considerations.

In the uncharged case $F_{\mu \nu}=m^2=0$, Christodoulou proved the validity of Conjecture \ref{trappedsurfaceconj}, which directly implies Statement \ref{trappedsurfaceconsequence} of Theorem \ref{Christotheorem} and is also the key ingredient of his proof of the Weak Cosmic Censorship Conjecture. 

One immediate consequence of Theorem \ref{roughversion} is the following fact, which was not previously recorded:

\begin{theo} \label{spaceliketheorem}
	Assume Conjecture \ref{trappedsurfaceconj}. Then Conjecture \ref{spacelikeconj} is true.
\end{theo}
The breakdown of weak null singularities -- a global property -- combined with Conjecture \ref{trappedsurfaceconj} -- a statement on $b_{\Gamma}$ -- implies there exists a non-central first singularity $p\in \mathcal{S}$, and that $b_{\Gamma}$ is a central first singularity. Therefore, the last step to obtain a full geometric understanding of spherical collapse is to prove Conjecture \ref{trappedsurfaceconj}, which would also imply the instability of naked singularities, a statement known as the Weak Cosmic Censorship conjecture, see section \ref{additionnal}.

\subsection{Additional related questions in charged gravitational collapse} \label{additionnal}

In this section, we give a brief review of the past works, conjectures and open problems related to the black hole interior. Ironically, the most prominent subsisting problem in gravitational collapse is related to the existence of singularities which form in the absence of a black hole. Such ``naked'' singularities are conjectured to be non generic. This statement -- the Weak Cosmic Censorship Conjecture -- can be formulated in modern terms as such: \begin{conjecture}[Weak Cosmic Censorship Conjecture for the Einstein--Maxwell--Klein--Gordon--Model] \label{WCC}
	Among all the data admissible from Theorem \ref{Kommemi}, there exists a generic sub-class for which $\mathcal{I}^+$ is complete.
\end{conjecture} Notice that in spherical symmetry, it can be proven that if the black hole region is non empty, then $\mathcal{I}^+$ is future-complete \cite{MihalisSStrapped,Kommemi}. One can also immediately show, c.f.\ \cite{Kommemi}, that Conjecture \ref{trappedsurfaceconj} implies Conjecture \ref{WCC}.


We now turn to another important problem, the Strong Cosmic Censorship Conjecture, which broadly states that General Relativity is a deterministic theory: \begin{conjecture}[Strong Cosmic Censorship Conjecture for the Einstein--Maxwell--Klein--Gordon--Model] \label{SCC}	Among all the data admissible from Theorem \ref{Kommemi}, there exists a generic sub-class for which the maximal globally hyperbolic development $(M,g)$ is future inextendible as a suitably regular Lorentzian manifold.
\end{conjecture}
The main obstruction to Strong Cosmic Censorship is the existence of Cauchy horizons, which can be smoothly extendible, e.g.\ for the Kerr stationary metric or for the Reissner--Nordstr\"{o}m static metric. Nevertheless, it was expected \underline{generic} dynamical Cauchy horizons feature a weak null singularity \cite{Ori} and, therefore, are $C^2$ inextendible. In the case of gravitational collapse additional obstructions related to the center appear, such as the existence of $\mathcal{CH}_{\Gamma}$, which is however empty generically if Conjecture \ref{trappedsurfaceconj} is true. Modulo these issues which are unrelated to weak null singularities, and assuming the quantitative stability of the black hole exterior, the author has obtained a version of Conjecture \ref{SCC}: \begin{thm}[$C^2$ inextendibility of the Cauchy horizon for the Einstein--Maxwell--Klein--Gordon model, \cite{Moi,Moi4}]
	Assume that Conjecture \ref{trappedsurfaceconj} is true. Then, under the assumptions of Theorem \ref{rough2}, the $C^2$ version of Conjecture \ref{SCC} holds i.e.\ $(M,g)$ is future inextendible as a $C^2$ Lorentzian manifold.
\end{thm}
We also mention the remarkable work of Luk and Oh \cite{JonathanStab,JonathanStabExt} who provide a comprehensive proof of the $C^2$ version of Strong Cosmic Censorship, for the Einstein--(uncharged)--scalar--field model. Note that in their case, the data are two-ended, thus there is no obstruction coming from the center of symmetry $\Gamma$, in contrast with our model. Note also that the decay of uncharged fields  in the exterior is well understood and governed by Price's law \cite{PriceLaw,Schlag,JonathanStabExt,Tataru}. However, Price's law does not apply to charged scalar fields, which obey more complicated dynamics, and decay is only known in the small charge case, see \cite{Moi2} for the proof of upper bounds that are sharp according to Conjecture \ref{chargedconj}.

We now return to the characterization of the interior boundary. In addition to the Weak  and Strong Cosmic Censorship conjectures, one problem is left unexplored: the causal character of $\mathcal{S}$. In particular, one can wonder whether $\mathcal{S}$ is space-like, in the sense of Definition \ref{spacelikedef}. The following ``space-like singularity'' conjecture appears reasonable:
\begin{conjecture}
	Among all the data admissible from Theorem \ref{Kommemi}, there exists a generic sub-class for which if the maximal future development has $\mathcal{Q}^+ \cap J^{-}(\mathcal{I}^+) \neq \emptyset $, then $\mathcal{S} \neq \emptyset$ is space-like.
\end{conjecture}
Note that, with the approach adopted in the present paper, we do not have any control over the causal character of $\mathcal{S}$, as we use a contradiction argument. Thus, it seems that a different approach is required to investigate this issue.

 We conclude this discussion with an interesting open problem: what happens in the interior of the black hole for the Einstein--Maxwell--Klein--Gordon equations in the presence of a positive cosmological constant? Based on the works \cite{Costa,Harvey}, the Cauchy horizon always exists, but is not weakly singular for a certain range of black hole parameters, i.e.\ the Hawking mass is finite. While it seems reasonable that our approach could be adapted to the cosmological setting, for parameters such that a weak null singularity is present, it would be interesting to see whether the Cauchy horizon may in some cases close off the space-time for parameters such that the Hawking mass is finite.

\subsection{Numerical and heuristic previous studies on $r=0$ singularities} \label{numerical}

The presence of $r=0$ singularities during the process of gravitational collapse  received a lot of attention from the numerical relativity community. In view of the work of Christodoulou, it is not the existence of those singularities which requires evidence, but the statement that they are generic, for models which allow for the formation of Cauchy horizons. However, it is difficult to validate generic statements numerically, as it requires to explore the whole moduli space of initial data.

The pioneering works of Brady and Smith \cite{BradySmith}, Burko and Ori \cite{BurkoOri} and Burko \cite{Burko} first provided evidence for the occurrence of $r=0$ singularities inside spherically symmetric Einstein--Maxwell-uncharged-scalar-field black holes. While the uncharged character of the scalar field does not significantly modify the asymptotic behavior, it forces the black hole to be \textit{two-ended} (section \ref{ChrisDaf}). In the two-ended case, the occurrence of $r=0$ singularities is \textbf{not} generic  \cite{Mihalisnospacelike}, see section \ref{twoended}. 

In contrast, to study the genericity of $r=0$ singularities in spherical collapse we must consider charged matter in \textit{one-ended} space-times, as in the present paper. One of the only numerical studies of charged gravitational collapse was carried out by Hod and Piran \cite{HodPiran2}, who considered the Einstein-charged-scalar-field system (i.e.\ \eqref{1.1}-\eqref{5.1} with $m^2=0$). For a particular choice of (global) initial data, including the center, they exhibited a Cauchy horizon, which is singular due to mass inflation, and at later times a singularity towards which the area-radius $r$ extends to zero. 

We also mention some heuristics of \cite{Chesler}, which attempt to argue in favor of the emergence of a $r=0$ singularity. It is not clear, however, what is the role of the center $\Gamma$ and of the Maxwell field in their work.

\subsection{Contrast with two-ended black holes, for charged/uncharged matter} \label{twoended}
\begin{figure}

	\begin{center}
		
		\includegraphics[width=89 mm, height=48 mm]{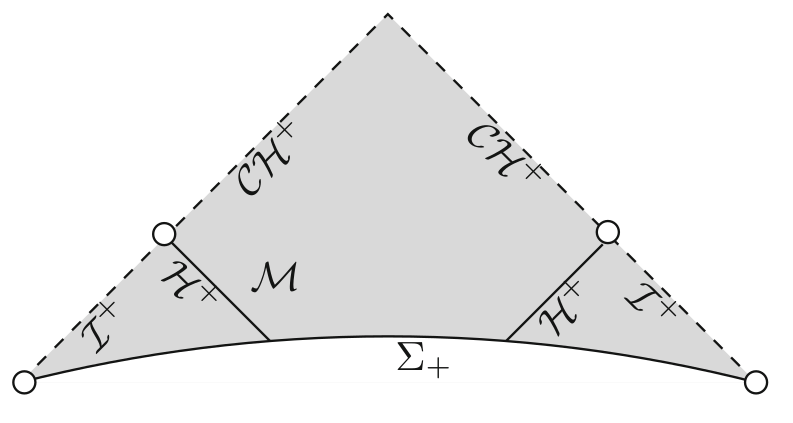}
		
	\end{center}
	
	\caption{Bifurcating Cauchy horizons in the two-ended case, for small data. Figure from \cite{Mihalisnospacelike}.}
	\label{Figure5}
\end{figure}

In this section, we describe what happens in the case of so-called \textit{two-ended} initial data that have  $\RR \times \mathbb{S}^2$ topology  (as opposed to the one-ended case $\RR^3$ considered in the present paper). Two-ended spacetimes are historically important, in that they have the same topology as Reissner--Nordstr\"{o}m or Kerr black holes. In such spacetimes, however, there is no center, i.e.\ $\Gamma=\emptyset$, and thus the two-ended data assumption \emph{does not model the global structure of gravitational collapse}. 
The two-ended setting is radically different from the one-ended case already for the Einstein--Maxwell-scalar-field model: in \cite{Mihalisnospacelike}, Dafermos proves that the two-ended analogue of Conjecture \ref{spacelikeconj} is \underline{false}:  there exists an open set of small, regular initial data for which there are no $r=0$ singularity and, an outgoing Cauchy horizon branches with an ingoing Cauchy horizon to close off the space-time at a bifurcation sphere as in Figure~\ref{Figure5}, just like for the Reissner--Nordstr\"{o}m solution. For those solutions constructed by Dafermos, the Hawking mass blows up everywhere, at least for a generic sub-class inside the open set of data. Thus, weak null singularities \textbf{do not necessarily} break down, in contrast with the one-ended case.\color{black}

We present a result which generalizes \cite{Mihalisnospacelike} to the more elaborate Einstein--Maxwell--Klein--Gordon model. The argument of \cite{Mihalisnospacelike} is easily transposable. We provide a sketch of the proof in Appendix \ref{twoendedA}.

\begin{thm} \label{twoendedtheorem}[Small scalar field data give rise a bifurcate Cauchy horizon in the two-ended case]
	Consider $(M,g,\phi,F)$ a solution of Einstein--Maxwell--Klein--Gordon system \eqref{1.1}, \eqref{2.1}, \eqref{3.1}, \eqref{4.1}, \eqref{5.1} arising from \underline{two-ended}, spherically symmetric regular initial data. 
	
	Assume moreover that the event horizon $\mathcal{H}^+$ settles quantitatively towards a sub-extremal Reissner--Nordstr\"{o}m event horizon and that the scalar field is \underline{small}, then there are no $r=0$ singularities: $\mathcal{S}=\emptyset$ and the Penrose diagram is given by Figure \ref{Figure5}, i.e.\ a bifurcate Cauchy horizon $\mathcal{CH}_{i_1^+} \cup \mathcal{CH}_{i_2^+}$ closes off the space-time. \end{thm}

This result proves that it is imperative to consider the global structure of the space-time for Theorem \ref{roughversion} to be valid; and, a fortiori, that there is no possible local approach to proving Conjecture \ref{spacelikeconj}. 

\subsection{Method of the proof and outline of the paper} \label{outline}

The proof of our main result, Theorem \ref{roughversion}, is by contradiction: we assume that the Penrose diagram is given by Figure~\ref{Fig3} (additionally, we may also have $\mathcal{S}_{i^+} \neq \emptyset$ as in Figure \ref{Figrect}, see Theorem~\ref{Kommemi} for a definition of  $\mathcal{S}_{i^+}$), where $\CH$ features a weak null singularity at an early time and we derive a contradiction from these two facts. 
In our approach, the main mechanism is the breakdown of weak null singularities;  the necessary occurrence of $r=0$ singularities (i.e.\ $\mathcal{S} \neq \emptyset$) is only obtained as an indirect consequence (assuming the absence of ``locally naked singularities'' i.e.\ $\mathcal{S}^1_{\Gamma} \cup \mathcal{CH}_{\Gamma} \cup  \mathcal{S}^2_{\Gamma} = \emptyset$), using the a priori result of Theorem \ref{Kommemi}.

Our methods are based on quantitative estimates involving the center of symmetry $\Gamma$ and the control of the Maxwell field by the Hawking mass. These estimates, based on the non-linear focusing properties of the Einstein equations in the presence of a weak null singularity, are proven on a causal rectangle with a top vertex $p=(u,v) \in \A$ ($\A$ being the apparent horizon),  a bottom vertex  $q \in \T$ ($\T$ being the trapped region) \color{black} and a left vertex on the center $\Gamma$, c.f.\ Figure \ref{Figrect}. As a consequence, we prove there exists a trapped ingoing null segment emanating from the top vertex $p$\color{black}.

Due to the geometry of the Penrose diagram given by Figure \ref{Fig3}, there exist such rectangles where $p=(u,v) \in \A$ is followed by an ingoing \underline{regular} segment. This is in contradiction with the consequence of the focusing estimates, which proves that the Penrose diagram of Figure \ref{Fig3}, where $\CH$ is a weak null singularity, was impossible in the first place.

\begin{figure}[H]
	
	\begin{center}
		
		\includegraphics[width=57 mm, height=50 mm]{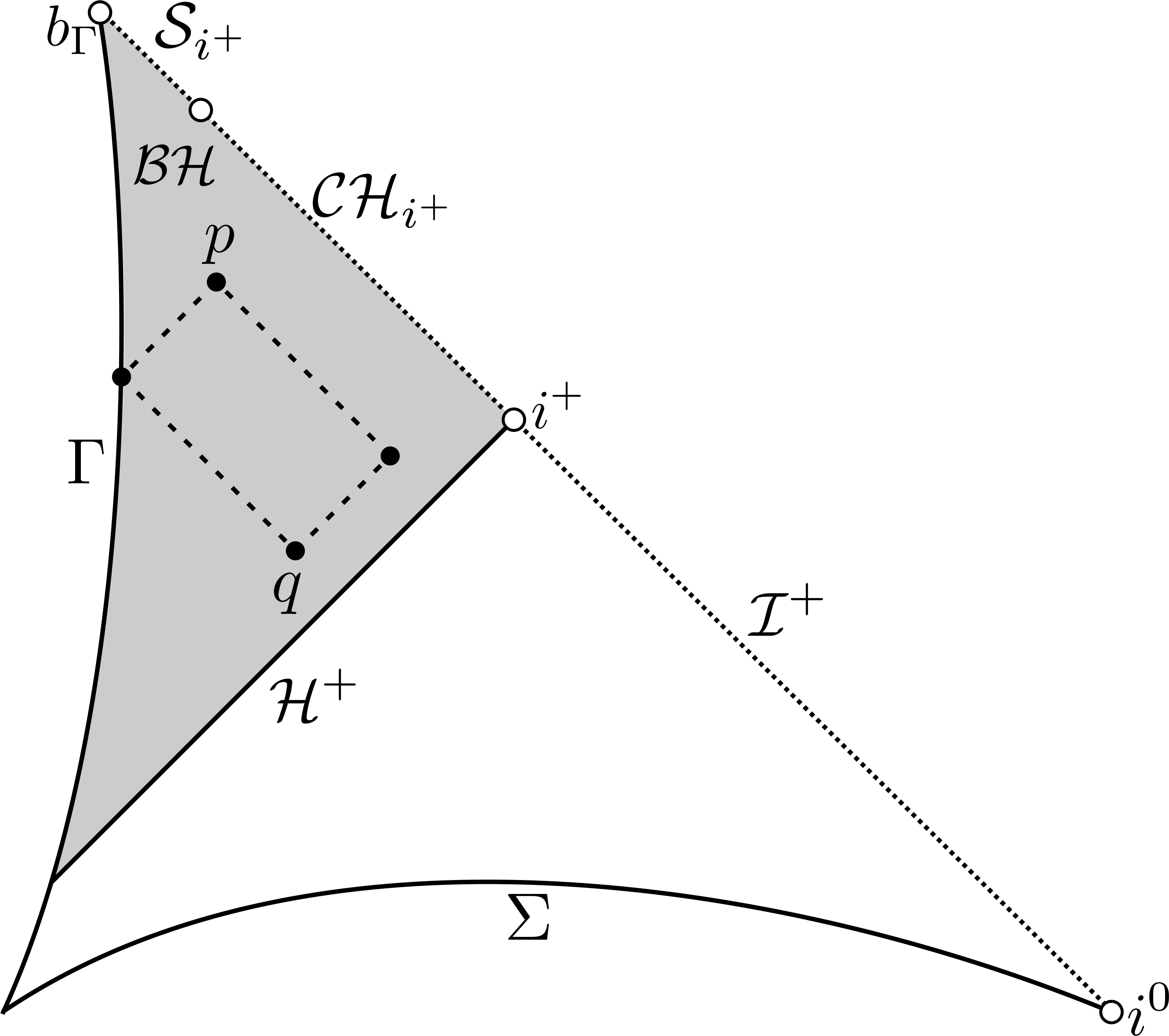}
		
	\end{center}
	
	\caption{The causal rectangle $J^+(q) \cap J^{-}(p)$ on which we prove focusing estimates, where $p \in \A$, $q \in \T$.}
	\label{Figrect}
\end{figure}

We outline the rest of the paper: in section \ref{geometricframework}, we lay out the geometric framework and the equations in double null coordinates. Then, we state precisely our results in section \ref{statement}. In section \ref{generaltheoremsection}, we provide the proof of Theorem \ref{roughversion}, i.e. Theorem \ref{maincorollary} and Theorem \ref{maintheorem}. In section \ref{finalsection}, we give the proof of Theorem \ref{rough2}, i.e. Theorem \ref{theoremevent}, using the main result of \cite{Moi4} that we recall. Finally, in Appendix \ref{twoendedA}, we prove Theorem \ref{twoendedtheorem}, implying that the two-ended case analogue of Conjecture \ref{spacelikeconj} is false.

\subsection*{Acknowledgments}I am very grateful to Mihalis Dafermos for the stimulating discussions we had about this problem, and for valuable comments on the manuscript. I warmly thank the anonymous referees for valuable suggestions. I also thank John Anderson for his interest and useful discussions.

\section{Geometric framework} \label{geometricframework}

The purpose of this section is to provide the precise setup, together with the definition of various geometric quantities, the coordinates and the equations that we will use throughout the paper.

\subsection{Spherically symmetric solution, as given by Theorem \ref{Kommemi}} \label{preliminary}
From Theorem \ref{Kommemi}, we obtain $(M,g,\phi,F)$, a regular solution of the system \eqref{1.1}, \eqref{2.1}, \eqref{3.1}, \eqref{4.1}, \eqref{5.1}, where $(M,g)$ is a Lorentzian manifold of dimension $3+1$, $\phi$ is a complex-valued function on $M$ and $F$ is a real-valued 2-form on $M$. 

$(M,g,\phi,F)$ is related to a quadruplet of scalar functions $ \{\Omega^2(u,v), r(u,v), \phi(u,v), Q(u,v)\}$, with $(u,v) \in \mathcal{Q} \subset \RR^{1+1}$ by \begin{equation} \label{gdef}\begin{split}
&g= g_{\mathcal{Q}}+ r^{2} \cdot  (d\theta^2 + \sin(\theta)^2 d\varphi^2)\\=\ & -\Omega^2(u,v) du dv+ r^{2}(u,v) \cdot  (d\theta^2 + \sin(\theta)^2 d\varphi^2), \end{split}
\end{equation} 
$$ F(u,v)= \frac{Q(u,v)}{2r^{2}(u,v)}  \Omega^{2}(u,v) du \wedge dv.$$  The domain $\mathcal{Q} \subset \RR^{2}$, called the Penrose diagram, is depicted in Figure \ref{Fig1} (for a choice of $(u,v)$, c.f.\ Remark \ref{choice}, such that $\mathcal{Q}$ is a bounded subset of $\RR^2$). The boundary components of Theorem \ref{Kommemi} can be identified with intervals in $\RR^2$ as $$i^+=\{(u_{-\infty}, v_{\infty})\},$$ $$\CH= (u_{-\infty}, u_{\CH}) \times \{v_{\infty}\},$$ $$ \mathcal{S}_{i^+}= [u_{\CH}, u(\CH \cup \mathcal{S}_{i^+} )) \times \{v_{\infty}\},$$ $$ b_{\Gamma}=\{(u(b_{\Gamma}), v(b_{\Gamma}))\}$$ $$\mathcal{S}_{\Gamma}^1= \{u(b_{\Gamma})\} \times (v(b_{\Gamma}), v(\mathcal{S}_{\Gamma}^1)),$$ $$\mathcal{CH}_{\Gamma}= \{u(b_{\Gamma})\} \times [v(\mathcal{S}_{\Gamma}^1), v( \mathcal{S}_{\Gamma}^1 \cup \mathcal{CH}_{\Gamma})),$$ $$\mathcal{S}_{\Gamma}^2= \{u(b_{\Gamma})\} \times [v(\mathcal{S}_{\Gamma}^1 \cup \mathcal{CH}_{\Gamma}), v(\mathcal{S}_{\Gamma}^1 \cup \mathcal{CH}_{\Gamma} \cup \mathcal{S}_{\Gamma}^2)),$$ where all these intervals are bounded, and possibly degenerate or empty (recall Theorem~\ref{Kommemi} for the precise definition of these boundary components). \begin{rmk}
	Note that the statement $\mathcal{S}  \neq \emptyset$ (obtained under the assumptions of Theorem \ref{spaceliketheorem}) is equivalent to \\$u(\CH \cup \mathcal{S}_{i^+})<u(b_{\Gamma})$. The statement $\mathcal{S}^1_{\Gamma} \cup \mathcal{CH}_{\Gamma} \cup  \mathcal{S}^2_{\Gamma} \cup  \mathcal{S}  \neq \emptyset$ appearing in Theorem \ref{roughversion} is equivalent to $v(b_{\Gamma})<v_{\infty}$.
\end{rmk}

One can now formulate the Einstein equations \eqref{1.1}, \eqref{2.1}, \eqref{3.1}, \eqref{4.1}, \eqref{5.1} as a system of non-linear PDEs on $\Omega^2$, $r$, $\phi$ and $Q$ expressed in the double null coordinate system $(u,v) \in \mathcal{Q}  $: \begin{equation}\label{Omega}
\partial_{u}\partial_{v} \log(\Omega^2)=-2\Re(D_{u} \phi \overline{D_{v}\phi})+\frac{ \Omega^{2}}{2r^{2}}+\frac{2\partial_{u}r\partial_{v}r}{r^{2}}- \frac{\Omega^{2}}{r^{4}} Q^2,
\end{equation} \begin{equation}\label{Radius}\partial_{u}\partial_{v}r =\frac{- \Omega^{2}}{4r}-\frac{\partial_{u}r\partial_{v}r}{r}
+\frac{ \Omega^{2}}{4r^{3}} Q^2 +  \frac{m^{2}r }{4} \Omega^2 |\phi|^{2} , \end{equation} 	\begin{equation}\label{Field}
D_{u} D_{v} \phi =-\frac{ \partial_{v}r \cdot D_{u}\phi}{r}-\frac{\partial_{u}r \cdot  D_{v}\phi}{r} +\frac{ iq_{0} Q \Omega^{2}}{4r^{2}} \phi
-\frac{ m^{2}\Omega^{2}}{4}\phi,\end{equation} 	\begin{equation} \label{chargeUEinstein}
\partial_u Q = -q_0 r^2 \Im( \phi \overline{ D_u \phi}),
\end{equation}	\begin{equation} \label{ChargeVEinstein}
\partial_v Q = q_0 r^2 \Im( \phi \overline{D_v \phi}),
\end{equation}
\begin{equation} \label{Rayching}\partial_{u}(\frac {\partial_{u}r}{\Omega^{2}})=\frac {-r}{\Omega^{2}}|  D_{u} \color{black}\phi|^{2}, \end{equation} 
\begin{equation} \label{RaychV}\partial_{v}(\frac {\partial_{v}r}{\Omega^{2}})=\frac {-r}{\Omega^{2}}|D_{v}\color{black}\phi|^{2},\end{equation} where the gauge derivative is defined by $D_{\mu}:= \partial_\mu+iq_0 A_{\mu}$, and the electromagnetic potential $A_{\mu}=A_u du + A_v dv$ satisfies 	\begin{equation} \partial_u A_v - \partial_v A_u = \frac{Q \Omega^2}{2r^2}.\end{equation}		
\begin{rmk}
	\eqref{Rayching} and \eqref{RaychV} are the Raychaudhuri equations, and we shall use this terminology in the paper.
\end{rmk} \begin{rmk} \label{choice}
The double null coordinates $(u,v)$ are not unique and can be re-parametrized as $du'=f_1(u) du$,  $dv'=f_2(v) dv$ for functions $f_1>0$, $f_2>0$. This procedure takes $\mathcal{Q}$ to a different domain  $\mathcal{Q}'$ but does not change the system of equations.
\end{rmk}

Subsequently, we define the Lorentzian gradient of $r$, and introduce the mass ratio $\mu$ by the formula $$ 1-\mu:=g_{\mathcal{Q}}(\nabla r,\nabla r),$$ where we recall that $g_{\mathcal{Q}}$ is the radial part of $g$ defined in \eqref{gdef}. We can also define the Hawking	
mass:

$$ \rho := \frac{\mu \cdot r}{2} =\frac{r}{2} \cdot(1- g_{\mathcal{Q}} (\nabla r, \nabla r )).$$	

Notice that the $(u,v)$ coordinate system, we have $g_{\mathcal{Q}} (\nabla r, \nabla r )= \frac{-4 \partial_u r \cdot \partial_v r}{\Omega^2}$. We can then define $\kappa$:

\begin{equation} \label{kappa}
\kappa = \frac{\partial_v r}{ 1-\frac{2\rho}{r}}=\frac{-\Omega^2}{4\partial_u r}\in \mathbb{R} \cup \{ \pm \infty\} .
\end{equation}

Now we introduce the modified mass $\varpi$ which involves the charge $Q$:

\begin{equation} \label{electromass}
\varpi := \rho + \frac{Q^2}{2r}= \frac{\mu r}{2} + \frac{Q^2}{2r} .
\end{equation}	An elementary computation relates the previously quantities : 	\begin{equation} \label{murelation}
1-\frac{2\rho}{r} = 1-\frac{2\varpi}{r}+\frac{Q^2}{r^2}=\frac{-4 \partial_u r \cdot \partial_v r}{\Omega^2}= \kappa^{-1} \cdot \partial_v r.
\end{equation}
On the sub-extremal Reissner--Nordstr\"{o}m spacetime $g_{RN}$ of mass $\varpi \equiv M>0$, charge $Q \equiv e$ with $0<|e|<M$, we denote $r_+=M+\sqrt{M^2-e^2}$, the area-radius of the event horizon $\mathcal{H}^+$, its surface gravity $2K_+:= \frac{2}{r^2_+}(M- \frac{e^2}{r_+})>0$, and $r_-=M-\sqrt{M^2-e^2}$, the radius of the Cauchy horizon, its surface gravity $2K_-:= \frac{2}{r^2_-}(M- \frac{e^2}{r_-})<0$.

In view of the definition of new quantities $\rho$, $\kappa$, one can derive new PDEs governing these quantities, as reformulations of the system previously written: the first two are variants of the ingoing Raychaudhuri equation \eqref{Rayching}:
\begin{equation} \label{RaychU}
\partial_u (\kappa^{-1}) =\frac {4r}{\Omega^{2}}|  D_{u} \color{black}\phi|^{2},
\end{equation} \begin{equation}\label{RaychU2}\partial_{u}(\log(\kappa^{-1}))=\frac {r}{|\partial_u r|}|  D_{u} \color{black}\phi|^{2}. \end{equation} Next, we can derive two transport equations for the Hawking mass using \eqref{Radius}, \eqref{RaychU}, \eqref{RaychV}: 
\begin{equation} \label{massUEinstein}
\partial_u \rho = \frac{-2r^2 \Omega^2}{\partial_v r}| D_u \phi |^2  +  \left(\frac{m^2}{2} r^2 |\phi|^2+ \frac{Q^2}{r^2} \right) \cdot \partial_u r ,
\end{equation}
\begin{equation} \label{massVEinstein} 
\partial_v \rho = \frac{r^2}{2\kappa}| D_v \phi |^2+ \left(\frac{m^2}{2} r^2 |\phi|^2+ \frac{Q^2}{r^2} \right) \cdot \partial_v r.
\end{equation}

Now we can reformulate our former equations to put them in a form that is more convenient to use. For instance, the Klein-Gordon wave equation \eqref{Field} can be expressed in different ways, using the commutation relation $[D_u,D_v]=\frac{ iq_{0} Q \Omega^{2}}{2r^{2}}$: \begin{equation}\label{Field2}
D_{u}(rD_{v} \phi) =-\partial_{v}r \cdot D_u\phi +\frac{ \Omega^{2} \cdot \phi}{4r} \cdot ( i q_{0} Q-m^2 r^2)
, \end{equation}\begin{equation}\label{Field3}
D_{v}(rD_{u} \phi) =-\partial_{u}r \cdot D_{v}\phi 
- \frac{ \Omega^{2} \cdot \phi}{4r} \cdot ( i q_{0} Q+m^2 r^2)
. \end{equation} We can also re-write  \eqref{Omega} and \eqref{Radius}: 	\begin{equation}\label{Omega3}
\partial_{u}\partial_{v} \log(r\Omega^2)=  \frac{ \Omega^{2}}{4r^{2}} \cdot \left(1-   \frac{3Q^2}{r^{2}} +m^{2}r^2 |\phi|^{2} - 8 r^2\Re( \frac{D_{u}\phi}{\Omega^2} \cdot  D_{v}\bar{\phi})  \right),
\end{equation}	\begin{equation} \label{Radius3}
-\partial_u  \partial_v (\frac{r^2}{2})	=\partial_u (-r \partial_v r) =	\partial_v (r | \partial_u r |) = \frac{\Omega^2}{4}\cdot (1- \frac{Q^2}{r^2}-m^2 r^2 |\phi|^2).
\end{equation}

\subsection{One-ended smooth solutions and regularity conditions on $\Gamma$} In view of \eqref{gdef}, we define $\Gamma:= \{ (u,v) \in \mathcal{Q}, r(u,v)=0 \}$, the center of symmetry. From Theorem \ref{Kommemi}, $\Gamma \neq \emptyset$ is time-like.

The smoothness of the solution $(M,g,\phi,F)$ imposes the following boundary conditions for the geometric quantities. \begin{equation} \label{reg1Gamma}
|g_{\mathcal{Q}} (\nabla r, \nabla r )_{|\Gamma}| <+\infty, \hskip 5 mm	|\phi|_{|\Gamma} <+\infty, \hskip 5 mm |F_{\mu \nu}|_{|\Gamma} <+\infty,  
\end{equation}

The reader can check that the regularity condition \eqref{reg1Gamma} imposes in particular the following boundary conditions, which are crucial in the present paper, and follow from Theorem \ref{Kommemi}: \begin{equation} \label{reg2Gamma}
\rho_{|\Gamma}=0 \hskip 5 mm	r\phi_{|\Gamma}=0, \hskip 5 mm Q_{|\Gamma} =0.
\end{equation}


Notice (see Figure \ref{Fig1}) also that every future directed ingoing ray must intersect $\Gamma$: for a fixed $v$, we denote $u_{\Gamma}(v)$, the $u$ coordinate of the intersection point: $(u_{\Gamma}(v),v) \in \Gamma$. Notice also that every past directed outgoing ray inside the black hole must 
intersect $\Gamma$: for a fixed $u$, we denote $v_{\Gamma}(u)$, the $v$ coordinate of the intersection point: $(u,v_{\Gamma}(u)) \in \Gamma$.


\subsection{Trapped region and apparent horizon}

We recall that it is assumed in Theorem \ref{Kommemi} that there is no anti-trapped surface in the initial data i.e.\ $\partial_u r_{|\Sigma_0} <0$ , hence $\partial_u r <0$ in the whole space-time, as a consequence of \eqref{RaychU}. In particular, $\kappa>0$ and $\partial_v r$ has the same sign as $	1-\frac{2\rho}{r} $.

We define the trapped region $\mathcal{T}$, the regular region $\R$ and the apparent horizon $\A$ as (c.f. \cite{Kommemi}):
\begin{enumerate}
	\item  \label{characttrapped}$(u,v) \in \T$ if and only if $\partial_v r(u,v)<0$ if and only if $1-\frac{2\rho(u,v)}{r(u,v)}<0$,
	\item $(u,v) \in \R$ if and only if $\partial_v r(u,v)>0$ if and only if $1-\frac{2\rho(u,v)}{r(u,v)}>0$,
	\item $(u,v) \in \A$ if and only if $\partial_v r(u,v)=0$ if and only if $1-\frac{2\rho(u,v)}{r(u,v)}=0$.
\end{enumerate}	
Note that the no anti-trapped surface assumption of Theorem \ref{Kommemi} implies, as $r_{|\Gamma}=0$, that $\Gamma \subset \R$.
\subsection{Double null coordinate choice}

We renormalize the coordinate $v$ by the condition $\partial_v r_{|\mathcal{H}^+}= 1-\frac{2\rho_{|\mathcal{H}^+}}{r_{|\mathcal{H}^+}}$, which is equivalent, by \eqref{murelation} to
\begin{align} \label{gauge1}
& \kappa_{|\mathcal{H}^+} \equiv 1, \\ & v(p)=0 \text{ where }  \{p\} = \mathcal{H}^+ \cap \Gamma.\label{gauge1.5}
\end{align}
Note that the second condition is simply a normalization of $v$ up to an additive constant, a choice which leaves \eqref{gauge1} invariant.\color{black}

As for the choice of $u$ coordinate, it is much less important since we will always write estimates which are independent of the $u$ coordinate choice. For concreteness, we introduce the following $u$-gauge \color{black} 
\begin{equation} \label{gauge2}
u_{\Gamma}(v)=v.
\end{equation}  Note that by \eqref{gauge1.5} we have \color{black} $\mathcal{H}^+= \{u=0\}$.

For convenience, we will keep the notations of section \ref{preliminary} for the boundary components. Note however that under the coordinate choice \ref{gauge1}, \ref{gauge2} we now have $u_{-\infty}=0$, $v_{\infty}=+\infty$,$u(b_{\Gamma})=+\infty$ \color{black} and $(u,v) \in \mathcal{Q}'$ where $\mathcal{Q}' \subset\RR^2$ is an \underline{unbounded} set in $\RR^2$.

\subsection{Electromagnetic gauge choice, and gauge invariant estimates}

The system of equations \eqref{1.1}, \eqref{2.1}, \eqref{3.1}, \eqref{4.1}, \eqref{5.1} is invariant under the gauge transformation : $$ \phi \rightarrow  e^{-i q_0 f } \phi ,$$
$$ A \rightarrow  A+ d f. $$
where $f$ is a smooth real-valued function. By an easy computation, one can show that the quantities $|\phi|$ and $|D_{\mu}\phi|$ are gauge invariant. We can then derive a gauge invariant estimate (see Lemma 2.1 in \cite{Moi2}): for all $u_1<u_2$, $v_1<v_2$: $$ |\phi|(u_2,v) \leq |\phi|(u_1,v)+ \int_{u_1}^{u_2} |D_u \phi|(u',v) du',$$
$$ |\phi|(u,v_2) \leq |\phi|(u,v_1)+ \int_{v_1}^{v_2} |D_v \phi|(u,v') dv'.$$

In the present paper, we only use such gauge invariant estimates and we will not involve $A_{\mu}$ in any computation.

\section{Statement of the main results} \label{statement}
We now give a precise statement of our theorems from section \ref{hypsection}, for space-times as in Theorem \ref{Kommemi}.
\subsection{Precise version of Theorem \ref{roughversion}}

We start with a precise version of Theorem \ref{roughversion}, the main result of the paper.
 \begin{thm} \label{maincorollary}
	For initial data as in Theorem \ref{Kommemi}, assume there exists an outgoing future cone emanating from $(u_1,v_1) \in \T$ and reaching $\CH$ on which $\phi$ and $Q$ obey the following upper bounds: for all $v \geq v_1$, \begin{equation} \label{corollaryassumption}
	|\phi|(u_1,v)+ |Q|(u_1,v) \leq C \cdot |\log(\rho)|,
	\end{equation} for some $C>0$ and the Hawking mass $\rho$ blows up \begin{equation} \label{corollaryassumption2}
	\lim_{v \rightarrow+\infty} \rho(u_1,v) =+\infty.
	\end{equation} Then $$\mathcal{S}^1_{\Gamma} \cup \mathcal{CH}_{\Gamma} \cup \mathcal{S}^2_{\Gamma} \cup  \mathcal{S}  \neq \emptyset.$$
\end{thm}

In reality, Theorem \ref{maincorollary} will be realized as a consequence of a more general theorem, for which we replace the mass blow-up by the more relaxed condition \eqref{integralkappafinite}, and we make the soft, but global assumption that there exists a trapped neighborhood of the Cauchy horizon, at least for sufficiently late times. While these assumptions can seem obscure at first, we prove that they are both satisfied if \eqref{corollaryassumption} and \eqref{corollaryassumption2} hold, so the following theorem is more general than Theorem \ref{maincorollary}:
\begin{thm} \label{maintheorem}
	For initial data as in Theorem \ref{Kommemi}, assume that $\CH \neq \emptyset$ and there exists $(u_1,v_1) \in \T$, with $(u_1,+\infty)\in~ \CH$ such that:
	\begin{equation} \label{integralkappafinite}
	\int_{v_1}^{+\infty} \kappa(u_1,v) (1+|\phi|^2(u_1,v)) dv	< +\infty.
	\end{equation} Denoting $\uend$, the $u$ coordinate of the  future \color{black} end-point of $\CH \cup \mathcal{S}_{i^+}$, assume that there exists $u_0 < \uend$ such that, for all $u_0  \leq u < \uend$, there exists $v(u)$ such that $(u,v(u)) \in \T$. Then $$\mathcal{S}^1_{\Gamma} \cup \mathcal{CH}_{\Gamma} \cup  \mathcal{S}^2_{\Gamma} \cup  \mathcal{S}  \neq \emptyset.$$

\end{thm} \begin{rmk} \label{nomassblowuprmk}
Theorem \ref{maintheorem} is a result of independent interest, as its assumptions are a priori uncorrelated with the blow up (or the boundedness) of the Hawking mass, therefore we may hope that they hold in various settings, in particular in the cosmological case where the blow up of the mass is not expected generically.
\end{rmk}

Because the integral of \eqref{integralkappafinite} is the crucial quantity governing the problem, we will first prove Theorem \ref{maintheorem} in section \ref{generaltheoremsection}, and then deduce Theorem \ref{maincorollary} in subsection \ref{propagationsection}, using the propagation of the Hawking mass blow up proven in \cite{Moi4}.

\subsection{Precise version of Theorem \ref{rough2}}
In the next theorem, we replace the integrability assumption \eqref{integralkappafinite} and the trapped neighborhood assumption of Theorem~\ref{maintheorem} by a  decay assumption on the event horizon $\mathcal{H}^+$  (at the expected rates, see section \ref{hypsection}).
\begin{thm} \label{theoremevent}
	We normalize $v$ by the gauge condition \eqref{gauge1}. For some $s>\frac{3}{4}$ we assume on $\mathcal{H}^+$, for all $v \geq v_0$ \begin{equation} \label{hypupperbound}
	|\phi|_{|\mathcal{H}^+}(v) + |D_v \phi|_{|\mathcal{H}^+}(v) \lesssim v^{-s},	\end{equation} \begin{equation} \label{hyplowerbound}
	\int_{v}^{+\infty} |D_v \phi|^2_{|\mathcal{H}^+}(v')dv' \gtrsim v^{-p},	\end{equation}
	for some $2s-1 \leq p \leq \min\{2s, 6s-3\}$. On the ingoing cone, we assume a red-shift estimate: \begin{equation} \label{RS} |D_u \phi|(u,v_0) \lesssim |\partial_u r|(u,v_0),
	\end{equation} for all $u \leq u_0$. Additionally, assume that a sub-extremal Reissner--Nordstr\"{o}m event horizon is approached, i.e. on $\mathcal{H}^+$ \begin{equation} \label{subextr}
	0<	 \limsup_{v \rightarrow +\infty} \frac{ |Q|_{|\mathcal{H}^+}(v)}{r_{|\mathcal{H}^+}(v)} <1.
	\end{equation}
	
	Then $$\mathcal{S}^1_{\Gamma} \cup \mathcal{CH}_{\Gamma} \cup \mathcal{S}^2_{\Gamma} \cup  \mathcal{S}  \neq \emptyset.$$

\end{thm}
\begin{rmk}
	The decay rates that we assume \eqref{hypupperbound}, \eqref{hyplowerbound} are conjectured to hold in the black hole exterior for generic Cauchy data, see section \ref{sectionrates}. Red-shift bounds such as \eqref{RS} are also conjectured to hold and reflect the fact that the event horizon $\mathcal{H}^+$ is a regular hyper-surface for the black hole metric and that $\phi$ is also regular across $\mathcal{H}^+$.
\end{rmk}


\section{Proof of Theorem \ref{maincorollary}  and Theorem \ref{maintheorem}} \label{generaltheoremsection}

\subsection{The strategy to prove  Theorem \ref{maincorollary}  and Theorem \ref{maintheorem}}
Theorem \ref{maincorollary}, with its stronger version Theorem \ref{maintheorem}, corresponds to Theorem \ref{roughversion}. Before starting their proof, we give in this sub-section an account of the strategy that we use in the rest of section \ref{generaltheoremsection}.

\subsubsection{The logic of the proof of Theorem \ref{maincorollary}  and Theorem \ref{maintheorem}}
In this section, we outline the proof that, under the assumptions of Theorem \ref{maincorollary}  and Theorem \ref{maintheorem}, $\mathcal{S}^1_{\Gamma} \cup \mathcal{CH}_{\Gamma} \cup  \mathcal{S}^2_{\Gamma} \cup  \mathcal{S}  \neq \emptyset$. The proof is divided in three steps, and follows a contradiction argument. Nonetheless, it relies on constructive focussing estimates (see section \ref{estimatesmethod}) which are valid independently and subsist, even after the contradiction has been established.

\begin{enumerate}
	\item \textit{Focusing properties on \textbf{a large class of} \color{black} causal rectangles with a vertex on the center, section \ref{keyestimatesection}} \label{stepone}
	
	We consider any causal rectangle of the form $[u_1,u] \times  [v_{\Gamma}(u),v]\subset \mathcal{Q}^+$, with $p=(u,v) \in \A$ and $q=(u_1,v_{\Gamma}(u)) \in \T$, as in Figure \ref{Figrect}.
	
	Assuming that the following focusing condition holds on the \\$\{u_1\} \times [v_{\Gamma}(u),v]$ side of the rectangle, for a small $\delta>0$, \begin{equation} \label{estimproof}
	\int_{v_{\Gamma}(u)}^{v} \kappa(u_1,v') \cdot (1+|\phi|^2(u_1,v')) dv' \leq \delta,
	\end{equation} we prove that the ingoing segment emanating from $(u,v)$ is trapped, i.e.\ there exists $ \eta>0$ with $(u,u+\eta) \times \{v\} \subset \T$.
	
	\item \textit{Construction of \textbf{one} rectangle with a vertex on the center, using the trapped neighborhood assumption, section \ref{halfdiamondsection}}
	
	For this step, we assume for now that there exists a neighborhood of the Cauchy horizon $\CH$ in the trapped region.
	
	Then, we work by contradiction and assume that  $\mathcal{S}^1_{\Gamma} \cup \mathcal{CH}_{\Gamma} \cup  \mathcal{S}^2_{\Gamma} \cup  \mathcal{S}  =\emptyset$. Then, we essentially establish that a connected component of $\A$ must terminate on $b_{\Gamma}$, via a soft argument using the trapped neighborhood of $\CH$ and the geometry resulting from $\mathcal{S}^1_{\Gamma} \cup \mathcal{CH}_{\Gamma} \cup  \mathcal{S}^2_{\Gamma} \cup  \mathcal{S}  =\emptyset$, c.f.\ the Penrose diagram of Figure \ref{Fig3}.
	
	As a result, we can construct one causal rectangle of the form $[u_1,u] \times  [v_{\Gamma}(u),v]$, with $(u,v) \in \A$ and $(u_1,v_{\Gamma}(u)) \in \T$, such that \eqref{estimproof} holds -- consequence of the blow up of the Hawking mass \eqref{roughassumption} -- and yet $(u,u+\epsilon) \times \{v\} \subset \R$ for some small $\epsilon>0$. This constitutes a \underline{contradiction}, using the result of Step \ref{stepone} that $(u,u+\eta) \times \{v\} \subset \T$ for $\eta>0$.
	
	\item \textit{Showing the existence of a trapped neighborhood of $\CH$ from the blow-up of the Hawking mass, section \ref{propagationsection}}
	
	In this final step, we prove, independently of the contradiction argument, and using \eqref{roughassumption} that there exists a trapped neighborhood of $\CH$. This follows from the propagation of the blow up of the Hawking mass over $\CH$ proven in \cite{Moi4}. Therefore, every outgoing cone over $\CH$ eventually satisfies $\frac{2\rho}{r}>1$, hence is trapped.
\end{enumerate}

\subsubsection{Novel focusing estimates, the key ingredient in the proof of  Theorem \ref{maincorollary}} \label{estimatesmethod}
The proof of Theorem \ref{maincorollary}  and Theorem \ref{maintheorem}, in particular step \ref{stepone}, relies on focusing estimates near the center, which are valid in any circumstance, independently of the proof of Theorem \ref{roughversion}. As a consequence of these estimates, the Hawking mass controls a very large flux of radiation. Thus, the charge, which is controlled by a smaller flux, is dominated by the Hawking mass.

\begin{enumerate}
	\item \textit{The Hawking mass $\rho$ controls an exponential flux of radiation} \label{step1}
	
	In the regular region $\mathcal{R}$, the combination of \eqref{RaychU2}\color{black} with \eqref{massVEinstein}  gives a focusing estimate  as follows (ignoring for now the $r$ weights).\begin{itemize}
		\item Integrating \eqref{RaychU2} in $u$ gives an estimate of the schematic form \begin{equation}
		\label{eq1}|\phi|^2 \lesssim \log(\kappa^{-1}).
		\end{equation} 
		\item Using \eqref{eq1}, \eqref{massVEinstein} gives an estimate of the schematic form	\begin{equation}\label{eq2}
 \partial_v \rho \gtrsim e^{|\phi|^2}  \cdot |D_v \phi|^2.
		\end{equation}
		\item  Integrating \eqref{eq2} in $v$ from\footnote{Note that the entire domain \color{black} of integration is contained in the regular region $\mathcal{R}$, see Section~\ref{preliminary} for more details.} the center where $\rho=0$, shows that $\rho$ controls  an exponential flux of radiation.
	\end{itemize}\color{black}

 We omitted various terms, some depending on the data on the outgoing cone $\{u_1\} \times [v_1,v_{max}]$, and $r$ factors.  
	The main novel ingredient through which we obtain this control is an a priori radiation flux estimate, see \eqref{ingoingphi}.
	
	\item \textit{The charge $Q$ is dominated by the Hawking mass $\rho$}
	
	By the Maxwell equation \eqref{ChargeVEinstein}, we also have an estimate on the charge $Q$,  using also \eqref{eq2} as follows \color{black} $$ |\partial_v Q| \lesssim |\phi| \cdot |D_v \phi| \ll \partial_v \rho,$$ which we integrate from the center where $Q=0$ and $\rho=0$. Thus, the charge is bounded  by the Hawking mass: $$ |Q|(u,v) \ll \rho(u,v).$$
	
	\item \textit{The consequence of focusing on the trapped region}
	
	In fact, the actual estimate which is obtained in step \ref{step1} is of the more specific form :
	$$ |Q|(u,v) \lesssim \rho(u,v) \cdot \int_{v_{\Gamma}(u)}^{+\infty} \kappa\cdot (1+|\phi|^2(u_1,v')) dv'.$$
	
	From the fact that $(u,v)$ is in the regular region/ apparent horizon, $\rho(u,v)	 \leq 2 r(u,v)$ (see Section~\ref{geometricframework}). Moreover, for any  small $\delta>0$, there exists $v$ such that \eqref{estimproof} is valid, using the blow up assumption \eqref{roughassumption}.
	
	
	Thus, combining the three inequalities, one can deduce $|Q|(u,v) \leq C(M,e,q_0,m^2) \cdot \delta \cdot r(u,v) < r(u,v),$ which then implies in the massless $m^2=0$ case, from \eqref{Radius3} and for $v$ large enough, with $(u,v)$ in the regular region: \begin{equation} \label{contradictoutline}
	\partial_u(r \partial_v r)(u,v) <0.
	\end{equation}	In turn, if $(u,v) \in \A$, this estimate directly implies that $(u,u+\eta) \times \{v\} \subset \T$ for some small $\eta>0$.  	\end{enumerate}

This is essentially the content of Theorem \ref{mainestimate}, with an estimate of the charge proven Lemma \ref{chargestimatelemma}. The massive case $ m^2 \neq 0$ can be also be treated: we then need an additional estimate for the massive term provided by Lemma \ref{massestimatelemma}.
\subsection{Two elementary calculus lemmata}
We start with two elementary computations, stated here for convenience. The first one is simply a one-dimensional functional inequality, which is important to handle the potential blow up in \eqref{ingoingphi}. The second one is a simple second order polynomial equation, which is useful to sort out the right smallness of $\delta$ required to apply Theorem \ref{mainestimate} and Lemma \ref{massestimatelemma}.
\begin{lem} \label{easylemma}
	
	Assume that for some $u>u_1$, $v > v_1$, $\{u\} \times [v_1,v] \subset \R \cup \A$ and $\{u_1\} \times [v_1,v] \subset \T \cup \A$. 
	
	Then, defining $r_1:=r(u_1,v_1)>0$, for all non-negative functions \color{black} $f$, we have the following estimate \begin{equation} \label{calculus1}
	\int_{v_1}^{v}  r(u,v') \log(\frac{r^{-1}(u,v')}{r^{-1}(u_1,v')}) f(u,v') dv'\leq 	r_1 \int_{v_1}^{v} f(u,v') dv'.\end{equation}

	\begin{proof}
		Since  $\{u_1\} \times [v_1,v] \subset \T \cup \A$, $r(u_1,v') \leq r_1$ for all $v_1 \leq v' \leq v$, and since $\partial_u r \leq 0$, $r(u,v) \leq r(u_1,v) \leq r_1$.
		Also, since  $\{u\} \times [v_1,v] \subset \R \cup \A$, for all $v_1 \leq v' \leq v$, $r(u,v') \leq r(u,v) \leq r_1 $.

		As a consequence of the inequalities, we have the estimate, defining $x(v')=\frac{r(u,v')}{r_1} \in (0,1]$: for all $v_1 \leq v' \leq v$: $$ r(u,v') \log(\frac{r^{-1}(u,v')}{r^{-1}(u_1,v')}) \leq r_1 \cdot  \frac{r(u,v')}{r_1} \log(\frac{r_1}{r(u,v')})= r_1 \cdot x \log(x^{-1}).$$
		
		The function $x \rightarrow x \log(x^{-1})$ is increasing on $(0,e^{-1})$ and decreasing on $(e^{-1},1]$, with a maximum at $x=e^{-1}$, whose value is $e^{-1} \leq 1$. This gives \eqref{calculus1} immediately.
	\end{proof}
\end{lem}

\begin{lem} \label{easylemma2}
	If $0<\delta <\frac{1}{32m^2 r_1 (1+ m^2r_1^2)}$, then	$$	m^2  \cdot \sqrt{ 8 r_1^3 \cdot \delta}  \cdot (1+ \sqrt{\frac{2\delta}{r_1}})  < \frac{1}{2}.$$
\end{lem}
\begin{proof} Set $y=\sqrt{\frac{2\delta}{r_1}}$: if $y_{\pm}$ are the roots of the polynomial second degree equation 	$y (1+y) =\frac{1}{4m^2 r_1^2}$, we must obtain $y< y_+$. We see by standard methods that $y_+= \frac{-1+\sqrt{1+ m^{-2} r_1^{-2}}}{2} = \frac{m^{-2} r_1^{-2}}{2(1+\sqrt{1+ m^{-2} r_1^{-2}})} \geq  \frac{m^{-2} r_1^{-2}}{4\sqrt{1+ m^{-2} r_1^{-2}}}.$
	
	Therefore, it is sufficient that $y=\sqrt{\frac{2\delta}{r_1}} <\frac{1}{4 m^2r_1^2\sqrt{1+ m^{-2} r^{-2}_1}}$, or equivalently $\delta <\frac{1}{32m^2 r_1 (1+ m^2r_1^2)}$.
\end{proof}
\subsection{An ingoing a priori estimate on $\phi$ in the entire space-time} \label{aprioriingoing}

Now, we establish a elementary focusing estimate, which relates the scalar field to the flux of ingoing radiation, quantified by $\kappa$, only using the Raychaudhuri equation. This estimate is important for section \ref{keyestimatesection}.
\begin{lem} For all $v \in \RR$, $u_1 \leq u \leq u_{\Gamma}(v)$, the following estimate is true:
	\begin{equation} \label{ingoingphi}
	|\phi|^2(u,v) \leq 	2|\phi|^2(u_1,v) +  2\log( \frac{\kappa^{-1}(u,v)}{\kappa^{-1}(u_1,v)}) \cdot  \log(\frac{r^{-1}(u,v)}{r^{-1}(u_1,v)}).
	\end{equation} 
	
\end{lem}
\begin{proof}

	We estimate $\phi(u,v)$ with respect to $\phi(u_1,v)$. Using Cauchy Schwarz and the ingoing Raychaudhuri equation \eqref{RaychU2}: \begin{equation*} \begin{split}
&	|\phi|(u,v)\\ \leq\ & 	|\phi|(u_1,v) + \int_{u_1}^{u}|D_u \phi|(u',v) du' \leq 	|\phi|(u_1,v) + (\int_{u_1}^{u} \frac{r |D_u \phi|^2}{|\partial_u r|}(u',v) du')^{\frac{1}{2}} (\int_{u_1}^{u} \frac{|\partial_u r|}{r}(u',v)du')^{\frac{1}{2}} \\ \leq\ & |\phi|(u_1,v) + ( \log( \frac{\kappa^{-1}(u,v)}{\kappa^{-1}(u_1,v)} ))^{\frac{1}{2}} \cdot  (\log(\frac{r^{-1}(u,v)}{r^{-1}(u_1,v)}))^{\frac{1}{2}},\end{split}
	\end{equation*} an estimate we can square, using $(a+b)^2 \leq 2a^2+2b^2$, which immediately gives \eqref{ingoingphi}.
\end{proof}
\subsection{The key estimate on a late rectangle with a vertex on the centre} \label{keyestimatesection}
Now, we get to the heart of the proof of Theorem \ref{maintheorem}: we establish focusing estimates, which will later be revealed, in section \ref{halfdiamondsection}, to be incompatible with the Cauchy horizon closing off the space-time. We will work on causal rectangles of the form $J^{-}(p) \cap J^+(q)$ as in Figure \ref{Figrect}, for $p=(u,v) \in \A$, and $q=(u_1,v_{\Gamma}(u)\color{black}) \in \T \cup \A$.
\begin{thm} \label{mainestimate}
	Let $(u,v) \in \A $. We define  $r_1:= r(u_1,v_{\Gamma}(u))$. We make the following assumptions on the causal rectangle $J^{-}(p) \cap J^{+}(q)= [u_1,u] \times  [v_{\Gamma}(u),v]$ for $p=(u,v)$ and $q=(u_1,v_{\Gamma}(u))$, as in Figure \ref{Figrect}:
	
	\begin{enumerate}
		
		\item $(u_1,v_{\Gamma}(u) )\in \T \cup \A$, therefore for all $v_{\Gamma}(u) \leq v' \leq v$, $(u_1,v')\in \T \cup \A$. 			\item For all $v_{\Gamma}(u) \leq v' < v$,  $(u,v') \in \R \cup \A $.
		
		\item We have the following (gauge invariant) estimate on the past outgoing boundary of the rectangle: \begin{equation} \label{mainassumption}
		\int_{v_{\Gamma}(u)}^{v} \kappa(u_1,v') \cdot (1+|\phi|^2(u_1,v')) dv' \leq \delta.
		\end{equation}
	\end{enumerate}
	
	Then, there exists $\delta_1(q_0,m^2,r_1)>0$, which we can choose to be $\delta_1 = \min \{ \frac{1}{4 q_0^2 r_1}, \frac{1}{32m^2 r_1 (1+ m^2r_1^2)}\}$ and such that, if $\delta  \leq \delta_1$, we have $\partial_u( r \partial_v r)(u,v)<0$. Therefore, a small future ingoing segment emanating from $(u,v)$ is included in the trapped region: there exists $\eta>0$ such that  $(u,u+\eta) \times \{v\} \subset \mathcal{T}$.
\end{thm}
We start with the main ingredient of the proof of Theorem \ref{mainestimate}: the control of the charge, in particular near the center. The following lemma is probably the most important result in the present paper:
\begin{lem} \label{chargestimatelemma}
	Under the assumption of Theorem \ref{mainestimate}, we have the following estimate at the top vertex $(u,v)$: \begin{equation} \label{Qcontrol}
	\frac{Q^2}{r^2}(u,v) \leq 2 r_1 \cdot  q_0^2 \cdot \delta < \frac{1}{2},
	\end{equation} where for the last inequality, we chose $\delta <   \frac{1}{4 r_1 \cdot q_0^2} $.
\end{lem}

\begin{proof}
	Using \eqref{ChargeVEinstein}, we have $|\partial_v Q| \leq |q_0| \cdot r^2  \cdot |\phi| \cdot |D_v \phi|$ which we integrate on $\{u\} \times [v_{\Gamma}(u),v]$, in the notations of Theorem \ref{mainestimate}. Using the fact that $Q(u,v_{\Gamma}(u))=0$, and the Cauchy-Schwarz inequality, we get \begin{equation} \label{Q1}
	|Q|(u,v) \leq |q_0| \int_{v_{\Gamma}(u)}^{v} r^2 |\phi| |D_v \phi|(u,v')dv' \leq |q_0| ( \int_{v_{\Gamma}(u)}^{v}\frac{ r^2}{2\kappa} |D_v \phi|^2)^{\frac{1}{2}} ( \int_{v_{\Gamma}(u)}^{v} 2r^2 \kappa |\phi|^2(u,v') dv')^{\frac{1}{2}}.
	\end{equation}  
	Since $\{u\} \times [v_{\Gamma}(u),v] \subset \R \cup \A$ by assumption, $\partial_v r(u,v') \geq 0$ for all $v' \in [v_{\Gamma}(u),v]$. Thus, by \eqref{massVEinstein}, we have $\frac{ r^2}{2\kappa} |D_v \phi|^2(u,v') \leq \partial_v \rho(u,v')$ for all $v' \in [v_{\Gamma}(u),v]$, hence, combining with \eqref{Q1} and using 
	 $\rho(u,v_{\Gamma}(u)) \geq 0$, we get: \begin{equation} \label{Q2}
	|Q|(u,v) \leq |q_0|  \int_{v_{\Gamma}(u)}^{v}r^2   |\phi|  |D_v \phi|(u,v') dv' \leq   |q_0|  \cdot \rho^{\frac{1}{2}}(u,v) \cdot  ( \int_{v_{\Gamma}(u)}^{v} 2r^2 \kappa |\phi|^2(u,v') dv')^{\frac{1}{2}}.
	\end{equation}  
	Now we estimate the term under the square-root, using \eqref{ingoingphi} to control $|\phi|^2(u,v')$: \begin{equation*} \begin{split}
&	\int_{v_{\Gamma}(u)}^{v} 2r^2 \kappa |\phi|^2(u,v') dv' \leq 2r(u,v) 	\int_{v_{\Gamma}(u)}^{v} r \kappa |\phi|^2(u,v') dv' \\ \leq\ &  4 r(u,v)\int_{v_{\Gamma}(u)}^{v} r(u,v') \kappa(u,v') |\phi|^2(u_1,v') dv' \\ +  &4 r(u,v)\int_{v_{\Gamma}(u)}^{v} r(u,v')\log(\frac{r^{-1}(u,v')}{r^{-1}(u_1,v')}) \kappa(u,v') \log( \frac{\kappa^{-1}(u,v')}{\kappa^{-1}(u_1,v')})dv' \\\leq\ & 4r_1 \cdot r(u,v) \cdot \left( \int_{v_{\Gamma}(u)}^{v} \kappa(u_1,v') |\phi|^2(u_1,v') dv'+ \int_{v_{\Gamma}(u)}^{v} \kappa(u,v') \cdot  \log( \frac{\kappa^{-1}(u,v')}{\kappa^{-1}(u_1,v')})dv'\right),
	\end{split}
	\end{equation*} where we used $r(u,v') \leq r(u,v)$ for the first inequality (since $\{u\} \times [v_{\Gamma}(u),v] \subset \R \cup \A$) and for the last inequality:  $r(u,v') \leq r_1$ for the first term (since $r(u,v) \leq r_1$, see the proof of Lemma \ref{easylemma}), and \eqref{calculus1} with $f(v')= \kappa(u,v') \log( \frac{\kappa^{-1}(u,v')}{\kappa^{-1}(u_1,v')})$ for the second term.
	
	Now, notice that $f(v')= \kappa(u_1,v') \cdot \frac{\log(x)}{x}$ for $x=\frac{ \kappa^{-1}(u,v')}{ \kappa^{-1}(u_1,v')}$. Then, we use the fact that $ \frac{\log(x)}{x} \leq e^{-1} \leq 1$ for any $x \in [1,+\infty)$, applied 
	to $x=  \frac{\kappa^{-1}(u,v')}{\kappa^{-1}(u_1,v')} \geq 1$. Therefore: \begin{equation} \label{ingoingphi2} \begin{split}
&	\int_{v_{\Gamma}(u)}^{v} 2r^2 \kappa |\phi|^2(u,v') dv' \leq 	2 r(u,v)\int_{v_{\Gamma}(u)}^{v} r \kappa |\phi|^2(u,v') dv' \\\leq\ &   4r_1 \cdot r(u,v) \cdot  \int_{v_{\Gamma}(u)}^{v} \kappa(u_1,v') (1+|\phi|^2(u_1,v')) dv' \leq 4 r_1 \cdot  \delta \cdot r(u,v),\end{split}
	\end{equation}where we used \eqref{mainassumption} in the last inequality. Thus, squaring \eqref{Q2}, and dividing by $r^2$ we get \begin{equation} \label{Q3}
	\frac{Q^2(u,v)}{r^2(u,v)} \leq q_0^2 \frac{(  \int_{v_{\Gamma}(u)}^{v}r^2   |\phi|  |D_v \phi|(u,v') dv')^2}{r^2(u,v)} \leq q_0^2 \cdot \rho(u,v) \cdot \frac{4r_1 \delta}{r(u,v)} =  2 r_1\cdot \delta \cdot  q_0^2  \cdot \frac{2\rho(u,v) }{r(u,v)} = 2 r_1\cdot \delta \cdot  q_0^2,
	\end{equation}
	
	where for the last inequality, we used $\frac{2\rho(u,v) }{r(u,v)}=1$, since $(u,v) \in \mathcal{A}$. This concludes the proof.
\end{proof}

In particular, Lemma \ref{chargestimatelemma} provides immediately a proof of Theorem \ref{mainestimate} in the massless case $m^2=0$, as in this case, we see by \eqref{Radius3} that $\partial_u(r \partial_v r)(u,v)=-\frac{\Omega^2}{4} (1-\frac{Q^2}{r^2})(u,v)< -\frac{\Omega^2(u,v)}{8}<0$, since $\Omega^2(u,v)>0$.

Now, we turn to the crucial estimate to handle the massive term, when $m^2 \neq 0$.

\begin{lem} \label{massestimatelemma}
	Under the assumption of Theorem \ref{mainestimate}, we have the following estimate in the top vertex $(u,v)$: \begin{equation}
	m^2 r^2 |\phi|^2(u,v) \leq  m^2  \cdot \sqrt{ 8 r_1^3 \cdot \delta}  \cdot (1+ \sqrt{\frac{2\delta}{r_1}})  < \frac{1}{2},
	\end{equation} where for the last inequality, we took $\delta < \frac{1}{32m^2 r_1 (1+ m^2r_1^2)}$.
\end{lem} \begin{proof}
We integrate  $\partial_v(r^2 |\phi|^2)$ on $\{u\} \times [v_{\Gamma}(u),v]$. Using the fact that $r^2 |\phi|^2(u,v_{\Gamma}(u))=0$, we get $$ r^2 |\phi|^2(u,v) \leq 	2\int_{v_{\Gamma}(u)}^{v} r^2 |\phi||D_v \phi|(u,v')dv'+ 2\int_{v_{\Gamma}(u)}^{v} r \partial_v r  |\phi|^2(u,v')dv' ,$$ where we used the identity $\partial_v(|\phi|^2)=2 \Re(\bar{\phi} D_v \phi)$. For the first term, we 
use \eqref{Q3} which we proved in Lemma \ref{chargestimatelemma} and for the second term, the inequality $\partial_v r(u,v') = (1-\frac{2\rho}{r}) \kappa(u,v') \leq \kappa(u,v')$ which holds for all $v_{\Gamma}(u)  \leq v' \leq v$ by \eqref{murelation} and because $\{u\} \times [v_{\Gamma}(u),v] \subset \R$, hence $\partial_v r(u,v') \geq 0$: thus, we get $$ r^2 |\phi|^2(u,v) \leq 2\sqrt{2} \cdot \sqrt{r_1} \cdot r(u,v) \cdot \sqrt{\delta}	+ 2\int_{v_{\Gamma}(u)}^{v} r \kappa |\phi|^2(u,v')dv' \leq  r_1^{\frac{3}{2}} \cdot \sqrt{8\delta}	+ 2\int_{v_{\Gamma}(u)}^{v} r \kappa |\phi|^2(u,v')dv',$$ 
and we already proved, see \eqref{ingoingphi2}, that $2\int_{v_{\Gamma}(u)}^{v} r \kappa |\phi|^2(u,v')dv' \leq 4 r_1 \cdot \delta$. Thus, combining everything: $$ m^2 r^2 |\phi|^2(u,v) \leq m^2  \cdot \sqrt{ 8 r_1^3 \cdot \delta}  \cdot (1+ \sqrt{\frac{2\delta}{r_1}})< \frac{1}{2},$$ where we chose $\delta<\frac{1}{32m^2 r_1 (1+ m^2r_1^2)}$ for the last estimate, by Lemma \ref{easylemma2}. This concludes the proof.
\end{proof}

In the case $m^2 \neq 0$, choosing  $\delta < \delta_1$ with $\delta_1$ defined in the statement of Theorem \ref{mainestimate}, the combination of Lemma \ref{chargestimatelemma} and Lemma \ref{massestimatelemma}  concludes the proof of Theorem \ref{mainestimate} since we have $\partial_u(r \partial_v r)(u,v)=-\frac{\Omega^2}{4} (1-\frac{Q^2}{r^2}-m^2 r^2 |\phi|^2)(u,v)<0$ by \eqref{Radius3}. To finish the proof, notice that $u \rightarrow r \partial_v r(u,v)$ is a $C^1$ function on $[u,u_{\Gamma}(v))$, and since $r \partial_v r(u,v)=0$ there exists $\eta>0$ such that for all $u' \in (u,u+\eta)$, $r \partial_v r(u',v)<0$, thus $\partial_v r(u',v)<0$, which proves that $(u,u+\eta) \times \{v\} \subset \T$.

\subsection{Existence of arbitrarily late half-diamonds in the regular region} \label{halfdiamondsection}
We start with \color{black} a geometric result: we construct half-diamonds with specific causal properties, and on which we will later apply the key estimate of section \ref{keyestimatesection}. To carry out this construction, we work by contradiction, assuming that $\mathcal{S}^1_{\Gamma} \cup \mathcal{CH}_{\Gamma} \cup \mathcal{S}^2_{\Gamma} \cup  \mathcal{S}  = \emptyset$ i.e.\ we assume by contradiction that the Cauchy horizon closes off the space-time in $b_{\Gamma}$ as depicted in Figure \ref{Fig3}.
\begin{prop}\label{halfdiamond} Assume, as in Theorem \ref{maintheorem}, that there exists $u_0 < \uend$ such that, for all \\$   u\in [u_0,\uend)$, there exists $v(u)$ such that $(u,v(u)) \in \T$. Assume also $\mathcal{S}^1_{\Gamma} \cup \mathcal{CH}_{\Gamma} \cup \mathcal{S}^2_{\Gamma} \cup  \mathcal{S}  = \emptyset$.
	
	Then, for all $v\geq v(u_0)$\color{black}, there exists $\uA\in[u_0, u_{\Gamma}(v))$ such that $(\uA,v) \in \A$ and \begin{enumerate}
		\item \label{appstat1}The future ingoing cone emanating from  $(\uA,v)$ lies in the regular region, i.e.\ $(\uA, u_{\Gamma}(v)]  \times \{v\} \subset \R$.
		\item \label{appstat2} The past outgoing cone emanating from  $(\uA,v)$ lies in the marginally regular region, i.e.\ $\{\uA\} \times [v_{\Gamma}(\uA), v]  \subset~ \R \cup ~\A$.
		\item \label{appstat3}\color{black}$(\uA,v)$ approaches $b_{\Gamma}$ as $v \rightarrow +\infty$: denoting $u(b_{\Gamma})$, the $u$ coordinate of $b_{\Gamma}$, we have \begin{equation} \label{apparentbgamma}
		\lim_{v \rightarrow +\infty} \uA= u(b_{\Gamma}).
		\end{equation} Therefore, we also have, given that $\Gamma$ is time-like, c.f.\ Figure \ref{Fig3}: \begin{equation} \label{apparentbgamma2}
		\lim_{v \rightarrow +\infty} v_{\Gamma}(\uA)= +\infty.
		\end{equation}
	\end{enumerate} 
\end{prop}

\begin{proof}  By assumption, $\CH \cup \mathcal{S}_{i^+} \neq \emptyset$ and since $\mathcal{S}^1_{\Gamma} \cup \mathcal{CH}_{\Gamma} \cup \mathcal{S}^2_{\Gamma} \cup  \mathcal{S}  = \emptyset$, then  $u(b_{\Gamma})=\uend$. By the trapped neighborhood assumption, for all $u_0 \leq u<u(b_{\Gamma})$, there exists $v(u)$  such that $(u,v(u)) \in \T$, and therefore, by the monotonicity of the Raychaudhuri equation \eqref{RaychV}, $\{u\} \times [v(u),+\infty) \subset \T$.
	
		\begin{figure}
		\begin{center}		
			\includegraphics[width=50mm, height=76.4 mm]{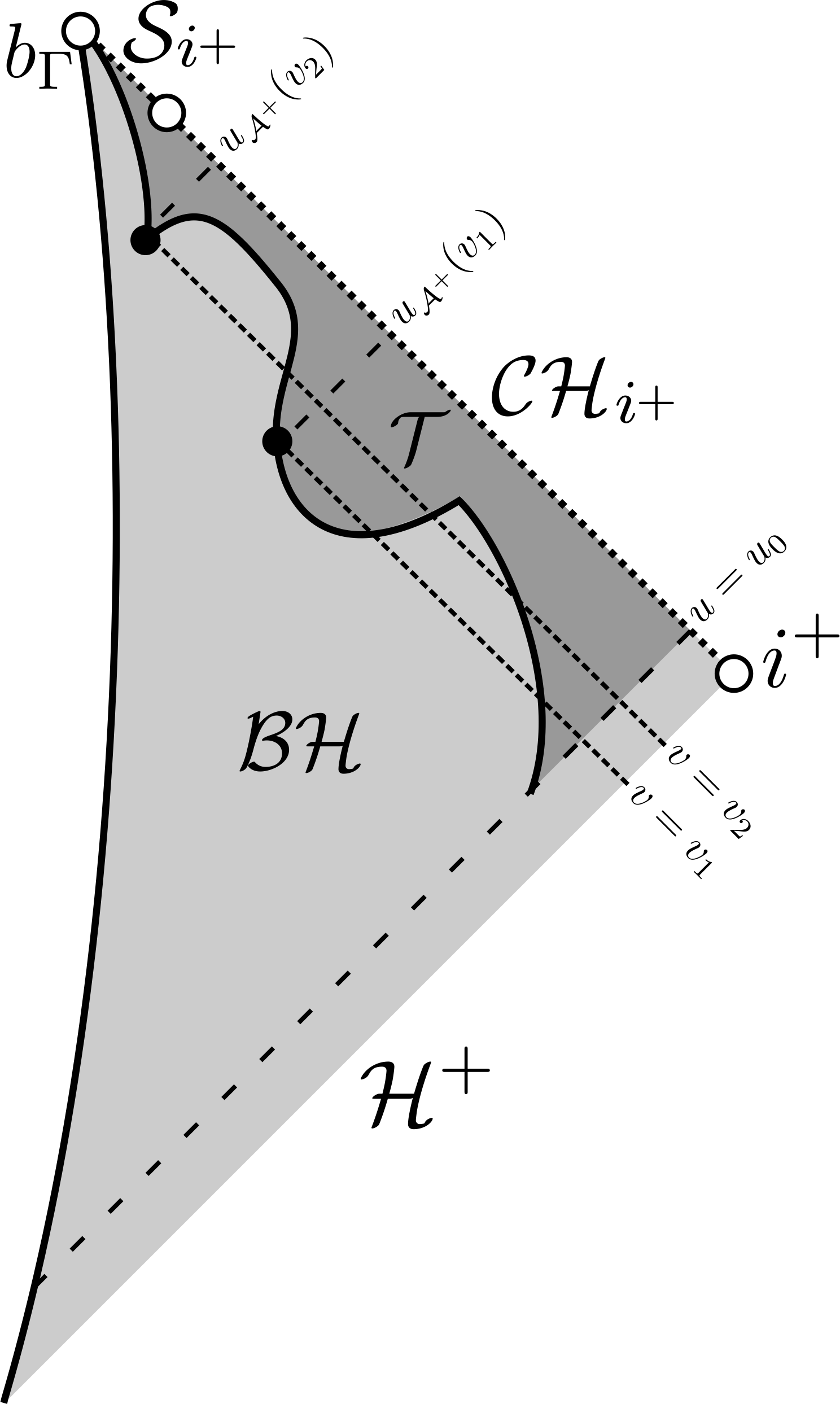}
		\end{center}
		\caption{Illustration of the proof of Proposition~\ref{halfdiamond} \color{black}} 
		\label{Fig5}
	\end{figure}
	
	For any $v \geq v(u_0)$, we know that $(u_0,v) \in \T$.  	
	Also $(u_{\Gamma}(v),v) \in \Gamma \subset \R$ so, because $u \rightarrow \partial_v r(u,v)$ is continuous on $[u_{0},u_{\Gamma}(v)]$, there exists $u'\in  (u_{0},u_{\Gamma}(v))$ such that $\partial_v r(u',v)=0$, by the intermediate value theorem. Hence $(u',v) \in \A$.
	Define $\uA:= \max \{ u' \in (u_{0},u_{\Gamma}(v)), \partial_v r(u',v)=0 \}$. By continuity, $\uA$ is well defined and $(\uA,v) \in \A$. By definition (and continuity again), for all $u'\in (\uA,u_{\Gamma}(v)]$, $\partial_v r(u',v)>0$ thus $(\uA,u_{\Gamma}(v)] \times \{v\} \subset \R$.
	
	$v\rightarrow \uA$ is a non-decreasing function with values on $[u_0,u(b_{\Gamma})]$: indeed, if $v< v'$ then by the monotonicity of the Raychaudhuri equation \eqref{RaychV}, $\partial_v r(\uA,v') \leq 0$, thus $\uA \leq  u_{\mathcal{A}^+}(v')  $. This means that $ \uA$ converges to a limit value $u_l \in [u_0,u(b_{\Gamma})]$ as $v\rightarrow+\infty$. We claim that $u_l=u(b_{\Gamma})$: if not, then the region $\{ u>u_l, v \geq  v(u_0)\} \subset \R$ by definition, obviously contradicting the trapped neighborhood assumption.
	
	Therefore, we chose $u=\uA$ and the claims \ref{appstat1} and \ref{appstat3} are satisfied, and only claim \ref{appstat2} remains. 
	
	For Claim \ref{appstat2}, notice $  \{\uA\}  \times [v_{\Gamma}(u),v] \subset \R \cup \A$, using again the monotonicity of the Raychaudhuri equation \eqref{RaychV}.
	
\end{proof} 
Now, combining Theorem \ref{mainestimate} and Proposition \ref{halfdiamond}, we prove Theorem \ref{maintheorem}:
\begin{cor} \label{above}
	Assume, as in Theorem \ref{maintheorem}, there exists $(u_1,v_1) \in \T$, with $(u_1,+\infty) \in \CH$, such that $$ \int_{v_1}^{+\infty} \kappa(u_1,v) (1+|\phi|^2(u_1,v)) dv < +\infty,$$ and that there exists $u_0 < \uend$ such that, for all $u_0  \leq u < \uend$, there exists $v(u)$ such that $(u,v(u)) \in \T$. Then $\mathcal{S}^1_{\Gamma} \cup \mathcal{CH}_{\Gamma} \cup  \mathcal{S}^2_{\Gamma} \cup  \mathcal{S}  \neq \emptyset$.
\end{cor}
\begin{proof}
	First, define $r_1=r(u_1,v_1)$ and then $\delta_1(q_0,m^2,r_1)$ by the expression given in the statement of Theorem \ref{mainestimate}. Using the integrability assumption, we see that there exists $v_0>v_1$ large enough such that, if $v\geq v_0$, $$ \int_{v}^{+\infty} \kappa(u_1,v) (1+|\phi|^2(u_1,v)) dv < \delta_1(q_0,m^2,r_1).$$
	
Then we proceed by contradiction: assume that $\mathcal{S}^1_{\Gamma} \cup \mathcal{CH}_{\Gamma} \cup \mathcal{S}^2_{\Gamma}  \cup  \mathcal{S}  = \emptyset$: then, for all $v> v_0$ large enough (recall that $v_0$ is defined right above)\color{black}, by Proposition \ref{halfdiamond}, there exists $u_{\A}\color{black}(v)>u_1$ such that $(u_{\A}\color{black}(v),v) \in \A$ , $[u, u_{\Gamma}(v)]  \times \{v\} \subset \R$, $\{u\} \times [v_{\Gamma}(u), v]  \subset \R \cup \A$ and $v_{\Gamma}(u)>v_1$, i.e.\ a causal rectangle $J^{-}(p) \cap J^{+}(q)$ as in Figure \ref{Figrect}, with $p=(u,v)$, $q =(u_1,v_{\Gamma}(u))$.

We choose such a $V>v_0>v_1$ large enough so that $v_{\Gamma}( u_{\A}\color{black}(V))>v_0>v_1$ (this is possible by the Claim \ref{appstat3} of Proposition  \ref{halfdiamond}). Therefore, we have in particular that $$ \int^{+\infty}_{v_{\Gamma}( u_{\A}\color{black}(V))} \kappa(u_1,v) (1+|\phi|^2(u_1,v)) dv < \delta_1(q_0,m^2,r_1).$$

	Notice also that $r(u_1,v_{\Gamma}(u_{\A}\color{black}(V))) \leq r_1$ as $v_{\Gamma}(u_{\A}\color{black}(V))>v_1$ and $\{u_1\} \times [v_1,\infty) \subset \T$, thus,\\ $\delta_1(q_0,m^2,r_1) \leq  \delta_1(q_0,m^2,r(u_1,v_{\Gamma}(u_{\A}\color{black}(V))))$ as $x \rightarrow \delta_1(q_0,m^2,x)$ is a decreasing function, which also implies that $$ \int^{+\infty}_{v_{\Gamma}( u_{\A}\color{black}(V))} \kappa(u_1,v) (1+|\phi|^2(u_1,v)) dv < \delta_1(q_0,m^2,r(u_1,v_{\Gamma}(u_{\A}\color{black}(V)))).$$
	
	Thus, the assumptions of Theorem \ref{mainestimate} are satisfied on the rectangle $[u_1,u] \times  [v_{\Gamma}(u),V\color{black}]$. 
	
	Hence $(u,u+\eta) \times \{V\color{black}\} \subset \T$ which contradicts $(u, u_{\Gamma}(V\color{black})]  \times \{V\color{black}\} \subset \R$ and the corollary is proven.
\end{proof}

The proof of Corollary~\ref{above} thus also concludes the proof of Theorem \ref{maintheorem}. 

\subsection{Proof of Theorem \ref{maincorollary} and propagation of the mass blow-up} \label{propagationsection}

We now start with the stronger assumptions of Theorem \ref{maincorollary}. We will simply prove that those assumptions imply the assumptions of Theorem \ref{maintheorem}, which we can then apply. For this, we need to use a result from \cite{Moi4}, which propagates the blow up of the Hawking mass, providing $\phi$ and $Q$ are controlled:
\begin{lem}[Propagation of the Hawking mass blow up, \cite{Moi4}] \label{massblowuplemma}
	If there exists $u_1 < u_{\CH} $ (recall the definition of $u_{\CH}$ from Section~\ref{preliminary}) \color{black} such that $$ \lim_{v \rightarrow +\infty} \rho(u_1,v)=+\infty,$$  $v_1 \in \RR$ large enough and a constant $D>0$ such that for all $v \geq v_1$: \begin{equation} \label{phiinitblowup}
	|\phi|^2(u_1,v)+ |Q|(u_1,v) \leq D \cdot |\log(\rho)|(u_1,v);
	\end{equation} 
	then, for all $ u_1 \leq u_2 < u_{\CH}\color{black}$, $$ \lim_{v \rightarrow +\infty} \rho(u_2,v)=+\infty.$$
	
	Thus, for all $u_1 \leq u < u_{\CH}\color{black}$, there exists $v(u)$ such that $(u,v(u)) \in \T$.	
\end{lem} \begin{proof} This statement is proven in \cite{Moi4} but we give a streamlined version of the argument for the benefits of the reader.

First, fix $u_2< u_{\CH}$. Note that, by definition (see Theorem~\ref{Kommemi}), for $v_1$ large enough, there exist constants $0<r_0^{-}<r_0^{+}$ such that $r_0^{-}<r(u,v)<r_0^{+}$ for all $(u,v)\in [u_1,u_2] \times [v_1,+\infty]$.  Fix  $\alpha=\frac{1}{4}$ \color{black} and $\eta_0>\eta>0$, we bootstrap in the region $[u_1,u_B]\times [v_1,+\infty)$  and $u_B \in [u_1,u_2]$: \begin{equation} \label{B1mblowup}
|\phi|^2+ Q^2 \leq \rho^{2\alpha},
\end{equation} 
\begin{equation} \label{B2mblowup}
\frac{-\partial_v r}{\Omega^2}  \geq \eta.
\end{equation} From the assumptions of the lemma and \eqref{RaychV}, it is clear that the set of spacetime points at which \eqref{B1mblowup} and \eqref{B2mblowup} are true is non-empty, for a small $\eta>0$  and taking $v_1$ larger if necessary. \color{black}

Then, using \eqref{massUEinstein} together with bootstrap \eqref{B1mblowup}, we have for some $C'(C,M,e,m^2,r_0^+,r_0^{-})>0$, 
$$ \partial_u \rho \geq \frac{2r^2 \Omega^2}{-\partial_v r}  |D_u \phi|^2 -C' \cdot \rho^{2\alpha} \cdot |\partial_u r|  \geq  -C' \cdot \rho^{2\alpha}  \cdot |\partial_u r|,$$ where for the last lower bound, we just used $\frac{-\partial_v r}{\Omega^2}\geq0$, as a soft consequence of \eqref{B2mblowup}. Since $0<\alpha<\frac{1}{2}$, it is clear that $$ \partial_u( \rho^{1-2\alpha})(u,v) \geq -(1-2\alpha) \cdot C' \cdot |\partial_u r|(u,v).$$

Thus, integrating we get that for all \color{black}  $(u,v)\in [u_1,u_B] \times [v_1,+\infty)$: \begin{align*} \rho^{1-2\alpha}(u,v) & \geq \rho^{1-2\alpha}(u_1,v) -  (1-2\alpha )C'  \cdot r(u_1,v)  \\ & \geq     \rho^{1-2\alpha}(u_1,v) -  (1-2\alpha )C'  \cdot r_+(M,e) \\ & \geq     \rho^{1-2\alpha}(u_1,v) -  C'', \end{align*} where $C''(C,M,e,m^2,r_0^+,r_0^{-})>0$. Note that to get the second inequality we have used $r(u_1,v) \leq r(0,v)= r_{|\mathcal{H}^+}(v) \leq r_+(M,e)$. From there we get $$  \rho(u,v)  \geq    \rho(u_1,v) \cdot [1 -  C'' \cdot  \rho^{2\alpha-1}(u_1,v)]^{\frac{1}{1-2\alpha}}. $$ Now remembering that $\lim_{v\rightarrow+\infty } \rho(u_1,v) =+\infty$ and Taylor-expanding gives us the following inequality, since $2\alpha-1<0$,  for some $C'''>0$
\color{black} \begin{equation} \label{massblowupest}
\rho(u,v) \geq \rho(u_1,v) \cdot (1-C'''\color{black} \cdot \rho^{2\alpha-1}(u_1,v)) .
\end{equation}

From this we obtain the blow up of the mass. Thus, there exists $v'_1>v_1$ such that for all $u_1 \leq u\leq u_B$, $v\geq v'_1$, $2\rho(u,v)>r_+$ thus $(u,v) \in \T$. Therefore, by \eqref{RaychV}, that $\frac{-\partial_v r}{\Omega^2}(u,v) \geq \eta_1>0$ for all $v\geq v_1'$ and $u_1 \leq u\leq u_B$, defining $\eta_1:=\sup_{u \in [u_1,u_2]}\frac{-\partial_v r}{\Omega^2}(u,v_1')$. Thus we retrieve bootstrap \eqref{B2mblowup} if $0<\eta<\eta_1$.

Now we need to retrieve bootstrap \eqref{B1mblowup}. For this, consider \eqref{massUEinstein} and write, under bootstrap \eqref{B1mblowup} and \eqref{B2mblowup} $$   \frac{r^2 \cdot |\partial_v r|}{2\Omega^2} \cdot |D_u \phi|^2(u,v) \leq C' \cdot|\partial_u r| \rho^{2 \alpha}(u,v)+ \partial_u \rho(u,v),$$ which is also equivalent, using \eqref{murelation} to \begin{equation}\label{E1}
\frac{ r|D_u \phi|^2}{2|\partial_u r|}(u,v) \leq \frac{ C'  \cdot|\partial_u r| \rho^{2 \alpha}(u,v)}{2\rho(u,v)-r(u,v)}+ \frac{\partial_u \rho(u,v)}{2\rho(u,v)-r(u,v)} \leq C'  \cdot  |\partial_u r|\ \rho^{-1+2 \alpha}(u,v)+ \partial_u \log(\rho)(u,v),
\end{equation}  where we have used $2\rho(u,v) -r(u,v) \geq \rho(u,v)$ on $[u_1,u_B] \times [v_1,+\infty]$ for $v_1$ large enough, since $\rho$ tends to $+\infty$.

Thus, we get, integrating \eqref{E1}, also using \eqref{massblowupest}: for all $u_1 \leq u\leq u_B$ \begin{equation} \label{E2}
\int_{u_1}^{u} \frac{ r|D_u \phi|^2}{2|\partial_u r|}(u',v) du' \lesssim   \rho^{-1+2 \alpha}(u_1,v)+  \log( \frac{\rho(u,v)}{\rho(u_1,v)}).
\end{equation} We can now use Cauchy-Schwarz to estimate $\phi$ as follows \color{black} \begin{equation}\begin{split}\label{E3}
 |\phi(u,v)- e^{-iq_0 \int_{u_1}^{u} A_u(u',v) du'}\phi(u_1,v)| \leq \int_{u_1}^{u} |D_u \phi|(u',v) du'\\ \leq (\int_{u_1}^{u} \frac{ r|D_u \phi|^2}{2|\partial_u r|}(u',v) du')^{\frac{1}{2}} (\int_{u_1}^{u} \frac{ 2|\partial_u r|}{r}(u',v) du')^{\frac{1}{2}}, \end{split}
\end{equation} which gives, combining \eqref{E2} and \eqref{E3} $$ |\phi|(u,v)\lesssim |\phi|(u_1,v)+   (\log( \frac{\rho(u,v)}{\rho(u_1,v)}) )^{\frac{1}{2}} \lesssim \log( \rho(u,v))^{\frac{1}{2}},$$ where we used \eqref{phiinitblowup} and \eqref{massblowupest}. A similar estimate is true for $Q$, using \eqref{chargeUEinstein}  so  we get for some constant $C>0$ (independent of $v_1$) and for all $u_1 \leq u \leq u_2$, $v\geq v_1$: $$ |\phi|^2(u,v) + Q^2(u,v) \leq C \log(\rho(u,v)) .$$ Since $\lim_{v \rightarrow +\infty} \log(\rho) \rho^{-2\alpha}(u,v)=0$, and taking $v_1$ larger if necessary, \color{black} we have  improved \color{black}  bootstrap \eqref{B1mblowup}  which concludes the proof. \color{black}


\end{proof}
Later, we will use this lemma to obtain a trapped neighborhood of $\CH$. Now we prove Theorem \ref{maincorollary}:

\begin{cor} \label{easycoro}
	For initial data as in Theorem \ref{Kommemi}, assume there exists an outgoing future cone emanating from $(u_1,v_1) \in \T$ with $ (u_1,+\infty) \in \CH$ and $C>0$ such that for all $v \geq v_1$: \begin{equation} \label{corollaryassumption3}
	|\phi|(u_1,v)+ |Q|(u_1,v) \leq C \cdot |\log(\rho)|,
	\end{equation} and the Hawking mass blows up:\begin{equation} \label{corollaryassumption4}
	\lim_{v \rightarrow+\infty} \rho(u_1,v) =+\infty,
	\end{equation} then \eqref{integralkappafinite} is satisfied and there exists $u_0 < \uend$ such that, for all $u_0  \leq u < \uend$, there exists $v(u)$ such that $(u,v(u)) \in \T$.
\end{cor} \begin{proof}
To prove \eqref{integralkappafinite}, we use \eqref{murelation} as $\kappa= \frac{ |\partial_v r|}{ \frac{2\rho}{r}-1}$. Then, we get, using \eqref{corollaryassumption3} that $$  \kappa  \cdot (1+|\phi|^2) \leq  |\partial_v r| \cdot \frac{1+C^2 \cdot [\log(\rho)]^2}{ \frac{2\rho}{r}-1},$$ and by \eqref{corollaryassumption4}, there exists $C'>0$ such that $$ \frac{1+C^2 \cdot [\log(\rho(u_1,v))]^2}{ \frac{2\rho(u_1,v)}{r}-1} \leq C',$$ for all $v$ large enough. Then, we have $$ \int_{v_1}^{+\infty}  \kappa  \cdot (1+|\phi|^2)(u_1,v')dv' \leq C' \int_{v_1}^{+\infty}( -\partial_v r) (u_1,v') dv' \leq C'  \cdot r(u_1,v_1) <+\infty.$$  

For the trapped neighborhood, we divide the proof in two cases: if $\mathcal{S}_{i^+}=\emptyset$ (recall the definition of $\mathcal{S}_{i^+}$ from Theorem~\ref{Kommemi}), then for all $u<u_{\CH}$, there exists $v(u)$ such that $(u,v(u)) \in \T$  by Lemma \ref{massblowuplemma} and $u_{\CH}=\uend$ so there is nothing more to do. If $\mathcal{S}_{i^+} \neq \emptyset$, for every $u_{\CH}<u<\uend$, $r(u,v) \rightarrow 0$ as $v \rightarrow +\infty$. Since for all $(u,v) \in \mathcal{Q}^+-\Gamma$, $r(u,v)>0$, $\{u\} \times [v_0,+\infty) \nsubseteq \R \cup \A$ for any $v_0>v_{\Gamma}(u)$ so there exists $v(u)$ such that $(u,v(u)) \in \T$ which concludes the proof, choosing $u_0>u_{\CH}$.
\end{proof}

Thus, the assumptions of Theorem \ref{maintheorem} are satisfied and this concludes the proof of Theorem \ref{maincorollary}.

\section{Proof of Theorem \ref{theoremevent}} \label{finalsection}

\subsection{The logic of the proof of Theorem \ref{theoremevent}}

In this section, we state three more steps to prove Theorem \ref{theoremevent}, using Theorem \ref{maincorollary}: thus, we only assume polynomial decay of the scalar field on the event horizon. This allows us to invoke a result from \cite{Moi4} which proves that there exists a neighborhood of $\CH \cup \mathcal{S}_{i^+}$ inside the trapped region and that, moreover, either \eqref{roughassumption} is satisfied, or $\CH$ is ``of static type'', meaning it is an isometric copy of the Reissner--Nordstr\"{o}m Cauchy horizon (step \ref{step4}). Then, we prove that a static type $\CH$ cannot be connected to $\mathcal{S}_{i+}$ (step \ref{step5}), or to $b_{\Gamma}$  (step \ref{step6}), which is sufficient to obtain the conclusion $\mathcal{S}^1_{\Gamma} \cup \mathcal{CH}_{\Gamma} \cup  \mathcal{S}^2_{\Gamma} \cup  \mathcal{S}  \neq\emptyset.$ (recall that all these boundary components and their notations were defined in Theorem~\ref{Kommemi}).		\begin{enumerate}
	\item \textit{Starting from assumptions on the event horizon: the dichotomy of \cite{Moi4}, section \ref{recallprevioussection}.} \label{step4}
	
	We use the main result of \cite{Moi4}: under the assumptions on the event horizon stated in Theorem \ref{theoremevent}, there exists a trapped neighborhood of $\CH \cup \mathcal{S}_{i^+}$ and one can classify the Cauchy horizon $\CH $: \begin{enumerate}
		\item \label{static} $\CH$ is of static type: then $\phi$, $r-r_-(M,e)$, $Q-e$, $\varpi-M$ extend to zero on $\CH$, meaning that $\CH$ is an isometric copy of the Reissner--Nordstr\"{o}m Cauchy horizon with parameters $M$ and $e$.
		
		\item  \label{dynamical} $\CH$ is of dynamical type: for all outgoing light cone $\{u\}\times [v_0,+\infty)$ to the future of the event horizon we have the following asymptotic estimates, in the gauge \eqref{gauge1} and for $v_0$  large enough: \begin{equation} \label{kappaoutline}
		\kappa(u,v) \lesssim C(u) \cdot e^{-\alpha(M,e) \cdot v}
		\end{equation} \begin{equation} \label{phioutline}
		|\phi|(u,v) \lesssim C(u) \cdot v.
		\end{equation}
		\item  \label{mixed} $\CH$ is of mixed type: there exists an outgoing light cone $\{u_1\}\times [v_0,+\infty)$ such that the asymptotic estimates \eqref{kappaoutline} and \eqref{phioutline} are true on any outgoing light cone to the future of  $\{u_1\}\times [v_0,+\infty)$.
	\end{enumerate}
	
	Now: either possibility \ref{dynamical} or possibility \ref{mixed} hold, in which case there exists a trapped outgoing cone on which \eqref{corollaryassumption}, \eqref{corollaryassumption2} are true thus $\mathcal{S}^1_{\Gamma} \cup \mathcal{CH}_{\Gamma} \cup  \mathcal{S}^2_{\Gamma} \cup  \mathcal{S}  \neq \emptyset$ by Theorem \ref{maincorollary}, or possibility \ref{static} holds, and the Cauchy horizon $\CH$ is of static type. The rest of the discussion will thus focus on the case where $\CH$ is of static type.
	
	In the rest of the proof, we establish, using radically different techniques, that in the latter case [possibility~\ref{static}], $\mathcal{S}^1_{\Gamma} \cup \mathcal{CH}_{\Gamma} \cup  \mathcal{S}^2_{\Gamma} \cup  \mathcal{S}  \neq \emptyset$ is still true. The core of the argument relies on the impossibility to connect a static Cauchy horizon to the other boundary components in presence when $\mathcal{S}^1_{\Gamma} \cup \mathcal{CH}_{\Gamma} \cup  \mathcal{S}^2_{\Gamma} \cup  \mathcal{S}  = \emptyset$.

	\item \textit{Impossibility to connect a static $\CH$ to a non-trivial $\mathcal{S}_{i+}$, section \ref{prooftheoremeventsection} as well.} \label{step5}
	
	We establish that, if $\CH$ is of static type, $\mathcal{S}_{i+}=\emptyset$: this follows from a more general result (c.f. Corollary \ref{Siplusempty} and Remark \ref{rcontinuityremark}) that we establish: if there exists a trapped neighborhood of the end-point of $\CH$ and $\mathcal{S}_{i+} \neq \emptyset$, then the area-radius $r$ must extend to a continuous function on that neighborhood, including at the end-point where $r=0$. The continuity of $r$ results from the wave equation \eqref{Radius3} with a regular right-hand-side (because we are in the trapped region) and the propagation of singularities. This is incompatible with a static Cauchy horizon, on which $r$ is a strictly positive constant: $r \equiv r_-(M,e)>0$.
	
	To complete the proof, the remaining task is to show, using a standard bootstrap method, that there exist a trapped neighborhood of the end-point of $\CH$, if $\CH$ is of static type (Proposition \ref{rigidboundsprop}).
	
	\item \textit{Impossibility to connect  a static $\CH$ to $b_{\Gamma}$, section \ref{prooftheoremeventsection}.} \label{step6}
	
	To conclude the argument, we must show that it is impossible to have $\mathcal{S}^1_{\Gamma} \cup \mathcal{CH}_{\Gamma} \cup  \mathcal{S}^2_{\Gamma} \cup  \mathcal{S}  \cup \mathcal{S}_{i+}=\emptyset$ and $\CH$ of static type. This follows directly from the estimates of Proposition \ref{rigidboundsprop}, which in particular conclude that there exists a space-time rectangle below the Cauchy horizon where $r$ is lower bounded, which is incompatible with the presence of a center $\Gamma$ where $r=0$.
\end{enumerate}

\subsection{The classification of the Cauchy horizon proven in \cite{Moi4}} \label{recallprevioussection}

Now, we set the preliminaries to the proof of Theorem \ref{theoremevent}. For this, we recall a theorem proven in \cite{Moi4}, under the same assumptions as Theorem \ref{theoremevent}: the main result is precisely the existence of a trapped neighborhood of $\CH$ together with some rigidity results, which are related to the blow up or the finiteness of the Hawking mass.

\begin{thm} [Classification of the Cauchy horizon,  Theorem 3.5 in \cite{Moi4}\color{black}]\label{classificationprevioustheorem}
	Under the assumptions of Theorem \ref{theoremevent}, then 
	$\CH \neq \emptyset$ and  there exists a neighborhood of $\CH$ inside the trapped region: for all $u<u_{\CH}$, there exists $v(u) \in \RR$ such that $\{u\} \times [v(u),+\infty) \subset \mathcal{T}$. Moreover, there is an alternative between two possibilities: \begin{enumerate}
		\item \label{alternative1} $\CH$ is of dynamical or mixed type in the language of \cite{Moi4} and as a consequence, there exists $u_1<u_{\CH}$, $\alpha(M,e)>0$, $C(u_1)>0$, $v_1 >0$ such that $(u_1,v_1) \in \T$ and for all $v \geq v_1$: \begin{equation}
		\kappa(u_1,v) \leq C \cdot e^{-\alpha v},
		\end{equation}
		\begin{equation}
		|\phi|(u_1,v) \leq C \cdot v^{1-s}.
		\end{equation} In particular, \eqref{integralkappafinite} and the other assumptions of Theorem \ref{maintheorem} are satisfied.
		
		\item  \label{alternative2} $\CH$ is of static type in the language of \cite{Moi4}, then $r-r_-(e,M)$, $\phi$, $D_u \phi$, $\varpi-M$, $Q-e$, $\partial_v \log(\Omega^2) -2K_-$ all extend continuously to $0$ on $\CH$, where $2K_-(M,e):= \frac{2}{r_-^2}[M-\frac{e^2}{r_-}]<0$ is the surface gravity of the Reissner--Nordstr\"{o}m Cauchy horizon\color{black}. Moreover, there exists $u_1<u_{\CH}$, $C(u_1)>1$, $v_1(u_1)>0$ such that $(u_1,v_1) \in \T$ and we have the following estimates, for all $v \geq v_1$: \begin{equation}\label{rigidity1}
		C^{-1} \cdot e^{2.01 K_- v}\leq 	\Omega^2(u_1,v) \leq  C \cdot e^{1.99K_- v}, 	\end{equation}  \begin{equation} \label{rigidity2}
		|\phi|(u_1,v)+|D_v \phi|(u_1,v) \leq  C \cdot v^{-s},
		\end{equation}  \begin{equation} \label{rigidity6}
		|\varpi(u_1,v)-M| +	|Q(u_1,v)-e|+	|r(u_1,v)-r_-(M,e)| \leq  C \cdot v^{1-2s},
		\end{equation}
		\begin{equation} \label{rigidity7}
		|\partial_v \log(\Omega^2)(u_1,v)-2K_-|\leq  C \cdot v^{1-2s},
		\end{equation}  	\begin{equation} \label{rigidity9}
		|\kappa^{-1}(u_1,v)-1| \leq  C \cdot v^{1-2s},
		\end{equation} 
		\begin{equation} \label{rigidity11}
		|\partial_v r|(u_1,v) \leq  C \cdot v^{-2s},	\end{equation}	\begin{equation} \label{rigidity12}
		\int_{v}^{+\infty}|\partial_v r|(u_1,v') dv' \geq  C^{-1} \cdot v^{-p}.	\end{equation}

	\end{enumerate}
\end{thm}
Now, we prove a corollary, which is an easy consequence of Theorem \ref{classificationprevioustheorem} but is not, strictly speaking, proven in \cite{Moi4}. This corollary proves the trapped neighborhood assumption in Theorem \ref{maintheorem}, under the assumptions of Theorem \ref{theoremevent}.
\begin{cor} \label{trappedcor}
	Under the assumptions of Theorem \ref{theoremevent}, there exists $u_0 < \uend$ such that, for all \\$u_0  \leq
	~ u <~ \uend$, there exists $v(u)$ such that $(u,v(u)) \in \T$.
\end{cor} \begin{proof}
If $\mathcal{S}_{i^+}=\emptyset$ then we apply Theorem~\ref{classificationprevioustheorem} and the result follows. If $\mathcal{S}_{i^+} \neq \emptyset$, then we argue exactly as in the proof of Corollary \ref{easycoro}:for every $u_{\CH}<u<\uend$, $r(u,v) \rightarrow 0$ as $v \rightarrow +\infty$, thus $\{u\} \times [v_0,+\infty) \nsubseteq \R \cup \A$ for any $v_0>v_{\Gamma}(u)$ so there exists $v(u)$ such that $(u,v(u)) \in \T$.
\end{proof}
\subsection{The proof of Theorem \ref{theoremevent}}  \label{prooftheoremeventsection}
We are now ready to prove Theorem \ref{theoremevent}, i.e.\ that $\mathcal{S}^1_{\Gamma} \cup \mathcal{CH}_{\Gamma} \cup  \mathcal{S}^2_{\Gamma} \cup  \mathcal{S}  \neq \emptyset$ under assumptions on the event horizon. First, we establish stability estimates, with data on a static Cauchy horizon and an outgoing cone. Those estimates are extremely strong, a consequence of the rigidity of Cauchy horizons of static type.  
\begin{prop} \label{rigidboundsprop}
	Assume that for some $u_1 < u_2 < u_{\CH}$, $r-r_-(e,M)$, $\phi$, $D_u \phi$, $\varpi-M$, $Q-e$ extend continuously to $0$ on $\CH \cap [u_1,u_2]$ and moreover estimates \eqref{rigidity1}, \eqref{rigidity2}, \eqref{rigidity6}, \eqref{rigidity7}, \eqref{rigidity9}, \eqref{rigidity11} are true on $\{u_1\} \times [v_1,+\infty) \subset \T$ for some $v_1>1$ and a constant $C>0$. Then there exists $v_1'=v_1'(C,M,e,q_0,m^2)>v_1$ such that the following are true: \begin{enumerate}
		\item The rectangle belongs to the space-time: $[u_1,u_2] \times [v_1',+\infty) \subset \mathcal{Q}^+$.
		\item The rectangle is trapped: $[u_1,u_2] \times [v_1',+\infty) \subset \T$ and the following estimate is true for all $u \in [u_1,u_2]$: 	\begin{equation} \label{lambda.main.est}
		-r\partial_v r(u,v_1') \geq A >0,
		\end{equation} where $A(C,M,e,q_0,m^2,p)>0$ is a constant independent of $u_1$ and $u_2$.

		\item Moreover, for all $\epsilon>0$, there exists $\tilde{v}(\epsilon,C,M,e,q_0,m^2)$ such that, if $v'_1> \tilde{v}$, the following estimates are true for all $(u,v) \in [u_1,u_2] \times [v_1',+\infty)$: 
		\begin{equation} \label{repsilon}
		|r(u,v)-r_-(M,e)| \leq \epsilon.
		\end{equation}
		
		Then, choosing $\epsilon \leq \frac{r_-}{2}$, there exists $v'_1(C,M,e,q_0,m^2)$ such that for all $u_1 <u_2<u_{\CH}$, $[u_1,u_2] \times [v_1',+\infty) \cap \Gamma = \emptyset$.
	\end{enumerate}
\end{prop}

\begin{proof}   To fix the ideas, without loss of generality, we fix $v_1=2$.
	In the proof, we consider characteristic initial data on $$\CH \cap \{u_1\leq u \leq u_2 \}	 \cup  \{u_1\} \times [2,+\infty).$$ satisfiying \eqref{rigidity1}-\eqref{rigidity12}. \\Whenever necessary, we will restrict the domain of evolution to $\{ v\geq v_1'\}$ where $v_1'(C,M,e,q_0,m^2,u_1) \geq v_1=2$ will be chosen adequately. The $v$-gauge will be chosen as \eqref{gauge1}. The $u$-gauge will be chosen later (note that \eqref{gauge2} does not make sense, as we consider a ``local problem'' on a spacetime rectangle that does not a priori intersect $\Gamma$).
	\color{black}
	
For some $D'(C,M,e,q_0,m^2)>0$ to be determined later, we formulate the following bootstrap assumptions: 	\color{black}  \begin{equation} \label{bootstrap2}
	|\phi|+	|D_v \phi| \leq 4 C  \cdot v^{-s},
	\end{equation} \begin{equation} \label{bootstrap3}
	|D_u \phi| \leq D'  \cdot \Omega^2 \cdot v^{-s}\color{black},\end{equation} \begin{equation} \label{bootstrap5}
	|Q|\leq 10 |e|,
	\end{equation}
	\begin{equation} \label{bootstrap7}
	r \geq \frac{r_-}{2},
	\end{equation} 
	\begin{equation} \label{bootstrap8}
	|\partial_v r| \leq 4 C \cdot  v^{-2s},
	\end{equation} 	 	\begin{equation}  \label{bootstrap9}
3 K_-(M,e)\leq  \partial_v \log(\Omega^2) \leq K_-(M,e),
\end{equation}
\begin{equation}  \label{bootstrap10}
 \frac{1}{2} \leq \kappa^{-1} \leq 2,
\end{equation}
\begin{equation}  \label{bootstrap11}
\int_{u_1}^{u_B} \Omega^2 du \leq e^{K_- v}.
\end{equation}
	
Note that \eqref{bootstrap2}-\eqref{bootstrap11} are all invariant with respect to the choice of $u$-gauge.
	We define the bootstrap set \begin{align*}
\begin{split}\mathcal{B} := &\{ u_B \geq u_1,\ [u_1,u_B] \times [v_1',+\infty) \subset \mathcal{Q}^+\\ &\text{ and }  \eqref{bootstrap2}-\eqref{bootstrap11} \text{ are satisfied for all } (u,v) \in [u_1,u_B] \times [v_1',+\infty) \}\end{split}
	\end{align*}
	
	\color{black}
	
	Note that, because of \eqref{rigidity1}- \eqref{rigidity11}, $\mathcal{B}$ \color{black} is non-empty, for $v'_1$ large enough. We will show that $u_B=u_2$, by proving that  \eqref{bootstrap2}-\eqref{bootstrap11} are satisfied with a smaller constant. 
	
	\color{black}
	Now integrate \eqref{RaychU}, using the boundedness of $r$ and bootstrap \eqref{bootstrap3}, \eqref{bootstrap10}\color{black}: $|\partial_u (\log\color{black}(\kappa^{-1}))| \lesssim  |\partial_u r| \cdot \color{black} v^{-2s} , $ which we can integrate to get, for $v'_1$ large enough \begin{equation} \label{kappaeq}
	|\kappa^{-1}-1| \leq C \cdot  v^{1-2s}+ \tilde{D}(M,e,q_0,m^2,C) \cdot v^{-2s} \leq  \frac{1}{100},
	\end{equation}  In particular, bootstrap \eqref{bootstrap10} is retrieved. Note that \eqref{kappaeq} also shows that $\kappa^{-1}$ extends to $1$ on $\CH$.
	
	\color{black}
	Integrate \eqref{Omega} in $u$, using the bound $|\partial_{u} \partial_v \log(\Omega^2)| \lesssim \Omega^2 $ which comes from \eqref{kappaeq}, and bootstraps \eqref{bootstrap3}, \eqref{bootstrap5}, \eqref{bootstrap7}, \eqref{bootstrap8}, we get  by \eqref{bootstrap11} \color{black} \begin{equation} \label{dvomega}
	|	\partial_v \log(\Omega^2)(u,v)-	\partial_v \log(\Omega^2)(u_1,v)\color{black}| \lesssim e^{K_- v}\color{black}\leq \frac{|K_-|}{100}, \end{equation} where we took $v'_1$ large enough. In particular,  bootstrap \eqref{bootstrap9} is retrieved, in view of \eqref{rigidity7} (again for $v'_1$ large enough). 
	
	Now define $\Omega^2_{CH}(u,v):= \frac{\Omega^2(u,v)}{\Omega^2(u_1,v)}$. Since \eqref{dvomega} shows that $\partial_v \log(\Omega^2_{CH}(u,v))$ is integrable, it means that $\Omega^2_{CH}(u,v)$ admits an extension $\Omega^2_{CH}(u)>0$ to $\CH$, and moreover \begin{equation}\label{om.est}
|	\log( \frac{\Omega^2_{CH}(u,v)}{\Omega^2_{CH}(u)})| \lesssim e^{K_- v}.
	\end{equation}

	\color{black}

	In parallel, note that by \eqref{Radius}, we have \begin{equation}\label{RadiusN}
\frac{-\partial_u (\partial_v r)}{-\partial_u r}= \frac{\kappa}{r} [ \frac{2\varpi}{r}-\frac{Q^2}{r^2}+ m^2 r^2 |\phi|^2].
	\end{equation}
	Recall that, by assumption, $(r,\varpi,Q,\phi)$ extends to $(r_-(M,e),M,e,0)$ on $\CH$, thus $\frac{-\partial_u (\partial_v r)}{-\partial_u r}$ extends to $2K_-=[\frac{2M}{r_-^2}-\frac{Q^2}{r_-^3}]<0$. Note that the above also shows that $\frac{-\partial_u r}{\Omega^2(u_1,v)}$ extends to a finite non-zero limit (namely $\Omega^2_{CH}(u)$) on $\CH$. Hence the above extension of $\frac{-\partial_u (\partial_v r)}{-\partial_u r}$ is also equivalent to the fact that the quantity 
$$-\partial_u (\frac{\partial_v r(u,v)-\partial_v r(u_1,v)}{\Omega^2(u_1,v)})$$ extends to $2K_-\cdot \Omega^2_{CH}(u)<0$, hence $  \Omega^2(u_1,2) \cdot\frac{\partial_v r(u,v)-\partial_v r(u_1,v)}{\Omega^2(u_1,v)}$ also extends to a finite  function $F(u)\geq0$ on $\CH$ (note that we multiplied $\frac{\partial_v r(u,v)-\partial_v r(u_1,v)}{\Omega^2(u_1,v)}$ by the constant  $\Omega^2(u_1,2)$ so that $F(u)$ is gauge-independent). Now choose the $u$-gauge to be \begin{equation}
\Omega^2_{CH}(u)=  F(u)+1.
\end{equation} In view of the above, we have that $\partial_u \log(\Omega^2_{CH})$ is constant on $\CH$, more precisely: \begin{equation}
\partial_u \log(\Omega^2_{CH}) \equiv 2K_- \cdot \Omega^2(u_1,2)<0,
\end{equation} from which we obtain the following estimate for some $D(u_1)>0$: \begin{align}\label{proper.time}
& \int_{u_1}^{u_B} \Omega^2_{CH}(u) du \leq \frac{\Omega^2_{CH}(u_1)}{2|K_-| \Omega^2(u_1,2)}= D(u_1),
\end{align}

Combining \eqref{proper.time} with \eqref{om.est}, we get, using \eqref{rigidity1} \begin{equation} \label{omega.estimate}
 \int_{u_1}^{u_B}\Omega^2(u,v) du  \leq  D(M,e,u_1) \cdot \Omega^2(u_1,v) \leq e^{1.99K_- v} \leq  \frac{1}{2} e^{K_- v},
\end{equation} for $v_1'(C,M,e,q_0,m^2,u_1)$ large enough, which retrieves bootstrap  \eqref{bootstrap11}.
	\color{black}

Then, using \eqref{kappaeq}, we write $|\partial_u r| = \frac{\kappa^{-1}}{4} \cdot \Omega^2 \leq  \Omega^2,$ which we integrate using \eqref{rigidity6}  and  \eqref{omega.estimate}\color{black} \begin{equation} \label{rmainestimate}
|r(u,v)-r_-| \leq 2 C \cdot v^{1-2s},
\end{equation} where we took $v'_1$  large enough: bootstrap \eqref{bootstrap7} is then retrieved.

Now we integrate \eqref{Field2} using \eqref{kappaeq}, bootstraps \eqref{bootstrap2}, \eqref{bootstrap3}, \eqref{bootstrap5}, the $r$ boundedness, and we get, for some constant $D(M,e,q_0,m^2,C\color{black})>0$ and using the upper bound $\partial_v \log(\Omega^2)  \leq 1.99 K_-$: $$|D_v (r D_u \phi)| \leq D \cdot \Omega^2  \cdot v^{-s}\color{black} = \frac{D  \cdot v^{-s}\color{black}}{-\partial_v \log(\Omega^2)} \cdot (-\partial_v \Omega^2) \leq \frac{D \cdot v^{-s}\color{black}}{1.99 |K_-|} \cdot (-\partial_v \Omega^2) ,$$ which we integrate from the future $\{v=+\infty\}$ where $D_u \phi_{|\CH} =0$ by assumption, also using bootstrap \eqref{bootstrap7}  and integration by parts\color{black}: \begin{equation} \label{Duestim}
|D_u \phi| \leq \frac{2  D \cdot v^{-s}\color{black}}{ 1.99 |K_-| \cdot r_-} \Omega^2.
\end{equation} This  retrieves \color{black} bootstrap \eqref{bootstrap3}, for $D' >  \frac{2 D}{ 1.99 |K_-| \cdot r_-} $. Additionally, one can integrate this estimate in $u$, using also \eqref{kappaeq}, \eqref{omega.estimate} \color{black}  to obtain: \begin{equation} \label{phiestim}
|\phi| \leq 2 C \cdot v^{-s},
\end{equation} where we took $v_1'$ large enough. Using \eqref{chargeUEinstein}, with bootstrap \eqref{bootstrap2} and \eqref{Duestim}, we have $|\partial_u Q| \lesssim  \Omega^2\color{black} $, thus  by  \eqref{omega.estimate} \color{black} \begin{equation} \label{Qmainestimate}
|Q(u,v)-e| \leq 2 C \cdot v^{1-2s},
\end{equation} where we took $v'_1$  large enough. This retrieves bootstrap \eqref{bootstrap5}.


Then using the upper and lower bounds on $r$, the upper bond on $Q$, \eqref{lambdamainestimate}, we see using \eqref{Field2} that \\$ |D_u (rD_v \phi)| \lesssim  \Omega^2 \cdot v^{-s}\color{black}$, which we can integrate, using \eqref{rigidity2}, \eqref{omega.estimate}\color{black}, and the largeness of $v'_1$ to get \begin{equation} \label{dvphi}
|D_v \phi|(u,v) \leq 2 C \cdot v^{-s},
\end{equation} which we combine with \eqref{phiestim} to close bootstrap \eqref{bootstrap2}.

	Now, we integrate \eqref{Radius3} in the $u$ direction, using \eqref{rigidity11}: \begin{equation} \label{lambdamainestimate} \begin{split}
| r \partial_v r (u,v) -r \partial_v r (u_1,v)| \leq e^{1.98 K_- v}, \\ | r \partial_v r (u,v)| \leq 2 r_- \cdot C  \cdot v^{-2s}+ e^{1.98 K_- v}\leq 4 r_- \cdot C  \cdot v^{-2s} .\end{split}
\end{equation} where we took $v'_1$ large enough. \eqref{lambdamainestimate} together with bootstrap \eqref{bootstrap7} allows us to retrieve bootstrap \eqref{bootstrap8}.

Thus, we have closed all the bootstraps and the first claim of the Proposition follows.

Finally, we estimate $\rho$ using \eqref{massUEinstein} as $|\partial_u \rho| \lesssim  \Omega^2 \color{black}$, using \eqref{lambdamainestimate}, \eqref{Duestim}, \eqref{kappaeq}, \eqref{phiestim}, \eqref{omega.estimate}\color{black} and the upper bounds on $r$ and $Q$: integrating, we obtain, for $v_1'$ large enough $$ |\rho(u,v) -\rho(u_1,v)| \leq e^{1.98 K_- v}.$$

Recall also that $\varpi = \rho+ \frac{Q^2}{2r}$, thus combining with \eqref{rmainestimate} and \eqref{Qmainestimate}, the upper-lower bounds on $r$ and $Q$, and \eqref{rigidity6}, we also get, taking $v_1'$ large enough \begin{equation} \label{mestimate}
|\varpi(u,v)-M| \leq 2C \cdot v^{1-2s}.
\end{equation}

From \eqref{rmainestimate}, \eqref{Qmainestimate}, \eqref{mestimate}, the third claim of the Proposition follows immediately.

Recall that $\{u_1\} \times [v_1,+\infty) \subset \T$ by assumption. By \eqref{lambdamainestimate} we have for all $(u,v)\in [u_1,u_2]\times[v_1',+\infty)$: $$-r\partial_v r(u,v) \geq r|\partial_vr|(u_1,v) -e^{1.98 K_- v}. $$

Now, note that by \eqref{rigidity12} there exists $C'>0$ and a sequence $v_n \rightarrow +\infty$ such that $$ r|\partial_v r |(u_1,v_n) \geq C' \cdot (v_n)^{-p-1}.$$ Therefore, there exists $N$ (independent of $u_1$ and $u_2$) large enough   (so that the term $C' \cdot (v_n)^{-p-1}$ dominates $e^{1.98 K_- v_n}$) such that for all $n \geq N$ and  for all $u \in [u_1,u_2]$: \begin{equation} \label{lambdanew}
-r\partial_v r (u,v_n) \geq \frac{ C' }{2}\cdot (v_n)^{-p-1}>0.
\end{equation}
Therefore, $[u_1,u_2] \times \{v_N\} \subset \T$. By the monotonicity of \eqref{RaychV}, $[u_1,u_2] \times [v_N,+\infty) \subset \T$. Therefore, choosing $v_1'>v_N$, the entire space-time rectangle is trapped. Choosing $v_1'=v_n$ for some $n \geq N$ gives \eqref{lambda.main.est} and the second claim of the Proposition follows.
	
\end{proof}
Using Proposition \ref{rigidboundsprop}, we now show that a static Cauchy horizon cannot be connected to $\mathcal{S}_{i^+}$.
\begin{cor} \label{Siplusempty}
	Assume that alternative \ref{alternative2} holds, i.e.\ $\CH$ is of static type. Then $\mathcal{S}_{i^+}=\emptyset$.
\end{cor}\begin{proof}

Suppose for the sake of contradiction that $\mathcal{S}_{i^+} \neq \emptyset$.

Set $\epsilon>0$ and take $u_1$ as in the statement of Theorem \ref{classificationprevioustheorem}, alternative \ref{alternative2}. By Proposition \ref{rigidboundsprop}, there exists $v_1'(\epsilon)$ such that for all $(u,v) \in [u_1,u_{\CH}) \times [v_1',+\infty)$,  \eqref{lambda.main.est} and \eqref{repsilon} hold for some $A>0$. Additionally, by continuity of the functions $u \rightarrow r(u,v)$ and $u \rightarrow -r\partial_v r(u,v)$ 
for every fixed $v \geq v_1'$, we also have that \eqref{repsilon} holds for all $(u,v) \in [u_1,u_{\CH}] \times [v_1',+\infty)$ and \eqref{lambda.main.est} holds for  all $u \in [u_1,u_{\CH}]$. In particular, taking $\epsilon$ smaller if necessary, we have $r(u,v) \geq \frac{r_-}{2}$ for all $(u,v) \in [u_1,u_{\CH}] \times [v_1',+\infty)$. Also, as a consequence of \eqref{lambda.main.est}, $(u_{\CH},v_1')\in \T$.

Since $\T$ is an open set in the topology of $\mathcal{Q}^+$, there exists $\eta>0$ such that $[u_{\CH}, u_{\CH}+\eta ] \times \{v_1'\} \subset \T$, thus by the monotonicity of the Raychaudhuri equation \eqref{RaychV} that  $[u_{\CH}, u_{\CH}+\eta ] \times [v_1',+\infty)\subset ~\T$. 

For every fixed $u \in [u_{\CH}, u_{\CH}+\eta ]$, $v \rightarrow r(u,v)$ has a limit $r_{CH}(u)\geq 0$ as $v \rightarrow +\infty$, by monotonicity. Now, we integrate \eqref{Radius3} in both $u$ and $v$ over  the space-time rectangle $[u_{\CH}, u_{\CH}+\eta ] \times [v_1',+\infty]$ to get a ``four points'' estimate: \begin{equation*} \begin{split}
\frac{r^2(u_{\CH}+\eta,v_1')}{2}+\frac{r^2_{CH}(u_{\CH})}{2}-\frac{r^2(u_{\CH},v_1')}{2}-\frac{r^2_{CH}(u_{\CH}+\eta)}{2}\\= \int_{u_{\CH}}^{u_{\CH}+\eta} \int_{v_1'}^{+\infty} \frac{\Omega^2}{4}(1-\frac{Q^2}{r^2}- m^2 r^2 |\phi|^2)(u',v') du' dv'  \leq \int_{u_{\CH}}^{u_{\CH}+\eta} \int_{v_1'}^{+\infty} \frac{\Omega^2(u',v')}{4} du' dv'   \\\leq \int_{u_{\CH}}^{u_{\CH}+\eta} \frac{\Omega^2(u',v_1')}{4 |\partial_v r|(u',v_1')} \int_{v_1'}^{+\infty}  |\partial_v r|(u',v')dv' du'   \leq \int_{u_{\CH}}^{u_{\CH}+\eta} \frac{\Omega^2(u',v_1') \cdot r(u',v_1')}{4 |\partial_v r|(u',v_1')}  du' \leq D \cdot \eta.
\end{split}
\end{equation*} where, from the second line to the third, we used the monotonicity of the Raychaudhuri equation \eqref{RaychV} and for the last inequality we set $D:= \| \frac{\Omega^2 \cdot r}{4 |\partial_v r|} \|_{L^{\infty}([u_{\CH}, u_{\CH}+\eta ] \times \{v_1'\})}<+\infty$ because the solution is smooth. Since, by monotonicity, we know $r_{CH}^2(u_{\CH}+\eta) \leq r_{CH}^2(u_{\CH})$, the previous estimate also implies $$ |r_{CH}^2(u_{\CH}+\eta)-r_{CH}^2(u_{\CH})| \leq 2 C \cdot \eta+ |r^2(u_{\CH}+\eta,v_1')-r^2(u_{\CH},v_1')|.$$ Thus, by continuity of the function $u \rightarrow  r^2(u,v_1')$ on  $[u_{\CH}, u_{\CH}+\eta ]$, we see that $$ \lim_{\eta \rightarrow 0} r_{CH}^2(u_{\CH}+\eta)=  r_{CH}^2(u_{\CH}) \geq \frac{r_-}{2}.$$

This contradicts that for $\eta$ small enough, $r_{CH}^2(u_{\CH}+\eta)=0$, by definition of $\mathcal{S}_{i^+}$. Thus, $\mathcal{S}_{i^+}=\emptyset$.
\end{proof}
\begin{rmk} \label{rcontinuityremark}
	Note that we proved a result of independent interest: if $\mathcal{S}_{i^+}\neq \emptyset$ and there exists a trapped neighborhood $\mathcal{V}$ of the end-point of $\CH$ in $\mathcal{Q}^+$, then $r$ extends to a continuous function on $\mathcal{V} \cap (\CH \cup \mathcal{S}_{i^+})$, in particular $r$ must be continuous at the end-point of $\CH$. Since we just used $\mathbb{T}_{v v} \geq 0$ and $\mathbb{T}_{u v} \geq 0$, the result is true in general for any matter model which satisfies the null energy condition.
\end{rmk}

We now conclude the proof of Theorem \ref{theoremevent}. Under the assumptions of Theorem \ref{theoremevent}: either alternative \ref{alternative1} holds or alternative \ref{alternative2} holds, by Theorem \ref{classificationprevioustheorem}. If alternative \ref{alternative1} holds, then the assumptions of Theorem \ref{maintheorem} are satisfied, also using Corollary \ref{trappedcor}, thus $\mathcal{S}^1_{\Gamma} \cup \mathcal{CH}_{\Gamma} \cup  \mathcal{S}^2_{\Gamma} \cup  \mathcal{S}  \neq \emptyset$. If alternative \ref{alternative2} holds, then $\mathcal{S}_{i^+}=\emptyset$ by Corollary \ref{Siplusempty}. Assume by contradiction that $\mathcal{S}^1_{\Gamma} \cup \mathcal{CH}_{\Gamma} \cup  \mathcal{S}^2_{\Gamma} \cup  \mathcal{S}  = \emptyset$. Then, $\CH$ closes off the space-time, and its endpoint is $b_{\Gamma}$, thus $u_{\CH}=u(b_{\Gamma})$. By Theorem \ref{classificationprevioustheorem}, alternative \ref{alternative2}, the assumptions of Proposition \ref{rigidboundsprop} are satisfied thus there exists $v_1' \in \RR$ such that $[u_1,u_{\CH}) \times [v_1',+\infty) \cap \Gamma = \emptyset$. This obviously leads to a contradiction, thus $\mathcal{S}^1_{\Gamma} \cup \mathcal{CH}_{\Gamma} \cup  \mathcal{S}^2_{\Gamma} \cup  \mathcal{S}  \neq \emptyset$. Therefore, the proof Theorem \ref{theoremevent} is complete.

\appendix
\section{Cauchy horizons can close-off a two-ended space-time} \label{twoendedA}
In this section, we give a brief sketch of the proof of Theorem \ref{twoendedtheorem}, mostly based on the strategy of \cite{Mihalisnospacelike}.
\begin{thm} 
	For some $s>\frac{1}{2}$ and $\epsilon>0$, we assume that the right event horizon $\mathcal{H}_1^+$ is a future-affine-complete outgoing cone and that  for all $v \geq v_0$ \begin{equation} 
	|\phi|_{|\mathcal{H}^+_1}(v) + |D_v \phi|_{|\mathcal{H}^+_1}(v) \leq \epsilon \cdot  v^{-s},	\end{equation} and that the following red-shift estimates hold:
	\begin{equation} |D_u \phi|(u,v_0) \leq \epsilon \cdot |\partial_u r|(u,v_0),
	\end{equation} for all $u \leq u_0$	 Additionally, assume that a sub-extremal Reissner--Nordstr\"{o}m black hole is approached on $\mathcal{H}^+_1$, i.e. \begin{equation}
	0<	 \limsup_{v \rightarrow +\infty} \frac{ |Q|_{|\mathcal{H}^+_1}(v)}{r_{|\mathcal{H}^+_1}(v)} <1.
	\end{equation}
	We make the same assumptions on the left event horizon $\mathcal{H}_2^+$. Then there exists an $\epsilon_0>0$ such that, if $\epsilon<\epsilon_0$, then the Penrose diagram is the one of Figure \ref{Figure5}, i.e.\ $\mathcal{CH}_1^+ \cup\mathcal{CH}_2^+$ close off the interior and $\mathcal{S}=\emptyset$.
\end{thm}
\begin{proof}
	The hard work, already established in \cite{Moi}, is to prove first the local result $\mathcal{CH}_1^+ \neq \emptyset$, $\mathcal{CH}_2^+ \neq \emptyset$ together with stability estimates. The Cauchy stability argument of \cite{Mihalisnospacelike} can be immediately adapted (Proposition 9.1 in \cite{Mihalisnospacelike}). Now, one must generalize Theorem 7 in \cite{Mihalisnospacelike} which states that $r$ is lower bounded, implying immediately the result. The only place where the Einstein equations are used in this result, is in equation (36), page 747, to establish that, in some region \begin{equation} \label{lastestimate}
	\partial_u \log(|\partial_v r|) \leq \frac{|\partial_u r|}{r}.
	\end{equation} This implication is not true in our model, as by \eqref{Radius}, $\partial_{u}\partial_{v}r =\frac{- \Omega^{2}}{4r}-\frac{\partial_{u}r\partial_{v}r}{r}
	+\frac{ \Omega^{2}}{4r^{3}} Q^2 +  \frac{m^{2}r }{4} \Omega^2 |\phi|^{2}$ thus, $$ \partial_u \log(|\partial_v r|) \leq \frac{|\partial_u r|}{r} +\frac{ \Omega^{2}}{|\partial_v r|} \cdot (\frac{1}{4r}-\frac{Q^2}{4r^3} -\frac{m^{2}r |\phi|^{2}}{2}),$$ and this estimate is valid in the region of interest, where $\partial_u r <0$ and $\partial_v r <0$. However, the new term does not create any problem as, by Cauchy stability, one still has $\frac{1}{4r}-\frac{Q^2}{4r^3} -\frac{m^{2}r |\phi|^{2}}{2}<0$, which the analog of equation (36). Therefore \eqref{lastestimate} is true and from there, one can follow the argument of \cite{Mihalisnospacelike} to the letter.
\end{proof}

\end{document}